\documentclass[a4paper]{amsart}
\usepackage[utf8]{inputenc}
\usepackage[a4paper,margin=2.5cm]{geometry}
\usepackage{amsmath,amssymb,amsthm,mathtools,abraces,comment}
\usepackage{textcomp}

\usepackage[dvipsnames]{xcolor}
\usepackage[hypertexnames=false,hyperfootnotes=false,colorlinks=true,linkcolor=blue,%
citecolor=purple,filecolor=magenta,urlcolor=cyan,unicode,linktocpage=true]{hyperref}

\newtheorem{theorem}{Theorem}[section]     
\newtheorem{proposition}[theorem]{Proposition} 
\newtheorem{lemma}[theorem]{Lemma} 
\newtheorem{corollary}[theorem]{Corollary} 
 
\newtheorem{conjecture}[theorem]{Conjecture} 
\theoremstyle{definition}
\newtheorem{definition}[theorem]{Definition}
\theoremstyle{remark}
\newtheorem{remark}[theorem]{Remark}
\newtheorem{example}[theorem]{Example}

\newcommand{\Mgn}{\overline{\mathcal{M}}_{g,n}}
\newcommand{\oM}{\overline{\mathcal{M}}}

\def\sfD{\mathsf{d}}
\def\sfQ{\mathsf{q}}
\def\sfB{{\mathsf{b}}}
\def\tsfE{\tilde{\mathsf{E}}}

\usepackage{filecontents}

\makeatletter
\newcommand{\subalign}[1]{%
	\vcenter{%
		\Let@ \restore@math@cr \default@tag
		\baselineskip\fontdimen10 \scriptfont\tw@
		\advance\baselineskip\fontdimen12 \scriptfont\tw@
		\lineskip\thr@@\fontdimen8 \scriptfont\thr@@
		\lineskiplimit\lineskip
		\ialign{\hfil$\m@th\scriptstyle##$&$\m@th\scriptstyle{}##$\crcr
			#1\crcr
		}%
	}
}
\makeatother

\title[Vanishing terms of the second Dubrovin--Zhang bracket]{Bi-Hamiltonian recursion, Liu--Pandharipande relations, and vanishing terms of the second Dubrovin--Zhang bracket}

\author{Francisco Hern{\'a}ndez Iglesias}

\address[F. Hern{\'a}ndez Iglesias]{Korteweg-de Vriesinstituut voor Wiskunde, 
	Universiteit van Amsterdam, Postbus 94248,
	1090GE Amsterdam, Nederland}
\email{f.hernandeziglesias@uva.nl}

\author{Sergey Shadrin}

\address[S. Shadrin]{Korteweg-de Vriesinstituut voor Wiskunde, 
	Universiteit van Amsterdam, Postbus 94248,
	1090GE Amsterdam, Nederland}
\email{s.shadrin@uva.nl}

\begin{document}

\begin{abstract} The Dubrovin--Zhang hierarchy is a Hamiltonian infinite-dimensional integrable system associated to a semi-simple cohomological field theory or, alternatively, to a semi-simple Dubrovin--Frobenius manifold. Under an extra assumption of homogeneity, Dubrovin and Zhang conjectured that there exists a second Poisson bracket that endows their hierarchy with a bi-Hamiltonian structure. More precisely, they gave a construction for the second bracket, but the polynomiality of its coefficients in the dispersion parameter expansion is yet to be proved.
	
In this paper we use the bi-Hamiltonian recursion and a set of relations in the tautological rings of the moduli spaces of curves derived by Liu and Pandharipande in order to analyze the second Poisson bracket of Dubrovin and Zhang. We give a new proof of a theorem of Dubrovin and Zhang that the coefficients of the dispersion parameter expansion of the second bracket are rational functions with prescribed singularities. We also prove that all terms in the expansion of the second bracket in the dispersion parameter that cannot be realized by polynomials because they have negative degree do vanish, thus partly confirming the conjecture of Dubrovin and Zhang. 
\end{abstract}

\maketitle

\tableofcontents

\section{Introduction} Dubrovin--Frobenius manifolds were first introduced in \cite{Dubrovin1996} as a way to study homogeneous solutions of the Witten-Dijkgraaf-Verlinde-Verlinde (WDVV) associativity equations \cite{WITTEN1990281,DIJKGRAAF199159}
%\begin{align}
%\frac{\partial^3 F}{\partial v^{\alpha_1} \partial v^{\alpha_2} \partial v^\alpha } \eta^{\alpha \beta} \frac{\partial^3 F}{\partial v^{\alpha_3} \partial v^{\alpha_4} \partial v^\beta } = \frac{\partial^3 F}{\partial v^{\alpha_2} \partial v^{\alpha_3} \partial v^\alpha } \eta^{\alpha \beta} \frac{\partial^3 F}{\partial v^{\alpha_1} \partial v^{\alpha_4} \partial v^\beta }
%\end{align}
in a coordinate-free way. In this work the relation between Dubrovin--Frobenius manifolds, topological field theories (TFTs) and integrable systems is first explored: namely, one can construct two flat 
%non-degenerate 
metrics related to a Dubrovin--Frobenius manifold, which give a pair of compatible Poisson brackets $(P, K)$ of hydrodynamic type. Then, starting from the Casimirs of $P$, $\bar{h}_{\alpha, -1} = \int \eta_{\alpha \beta} v^\beta dx$, and under the assumption that there are no common Casimirs, one can apply a bi-Hamiltonian recursion algorithm~\cite{magri} to obtain a dispersionless hierarchy of the form
\begin{align}
\frac{\partial v^\alpha}{\partial t^\beta_q} = P^{\alpha \gamma} \frac{\delta \bar{h}_{\beta, q}}{\delta v^\gamma}.
\end{align}

In \cite{DZ01}, Dubrovin and Zhang further explore the relationship between Dubrovin--Frobenius manifolds and integrable systems and deform this hierarchy via a quasi-Miura transformation
\begin{align}\label{eq:quasiMiura}
w^\alpha = v^\alpha + \sum_{g=1}^{\infty} \epsilon^{2g} Q^\alpha_g (v, v_1, \dots, v_{3g})
\end{align}
given by weighted homogeneous differential rational functions $Q^\alpha_g$ to obtain the full dispersive hierarchy, which is known as the Dubrovin--Zhang (DZ) hierarchy. They conjecture in \cite{DZ01} that the transformed equations, Hamiltonians and brackets are differential polynomials in the coordinates $w^\alpha$.

In \cite{BPS122, BPS12} this conjecture is partially proved in a more general setting: the  DZ hierarchy is constructed  from a semi-simple cohomological field theory (CohFT), without an assumption of homogeneity.
%which corresponds to a Dubrovin--Frobenius manifold if and only if it is conformal. 
%Using the Givental group action (\cite{Giv01}) on the space of tame partition functions, 
It is proved in \emph{op.~cit.} that the equations, Hamiltonians, tau structure, and first bracket of the DZ hierarchy are polynomial.

The main goal of this paper is to analyze the second Poisson bracket of the Dubrovin--Zhang hierarchy. We start with a conformal semi-simple cohomological field theory, thus the construction of Dubrovin applied to the underlying Dubrovin--Frobenius manifold gives the second Poisson bracket in the dispersionless limit, and the quasi-Miura transformation~\eqref{eq:quasiMiura} produces a possibly singular Poisson structure. We have two tools to analyze it: the bi-Hamiltonian recursion and the tautological relations in cohomology of the moduli spaces of curves or, more precisely, the differential equations that they imply on various structures of the Dubrovin--Zhang hierarchies.

The bi-Hamiltonian recursion appears to be sufficient to uniquely determine the second Dubrovin--Zhang bracket, and we also use it to give a new proof of a structural result of Dubrovin and Zhang on its possible singularities. Bi-Hamiltonian recursion also implies that the constant term of the second Dubrovin--Zhang bracket is a differential polynomial. 

As a source of suitable tautological relations we use the work of Liu and Pandharipande~\cite{Liu-Pand}. The relations that they derive there appear to be exactly enough to prove the vanishing of all terms in the second Dubrovin--Zhang bracket whose standard degree is negative. Remarkably, the dimensional inequalities of the Liu--Pandharipande relations match exactly the standard degree count for the terms of the second Dubrovin--Zhang bracket in the equations that we derive from the bi-Hamiltonian recursion, so the Liu--Pandharipande relations say nothing about the non-negative standard degree terms of the second bracket. 

\begin{comment}
The main result of this paper is the proof of the remaining conjecture, namely
\begin{conjecture} \label{conj:mainconjecture}
	The second Poisson structure of the Dubrovin-Zhang hierarchy is polynomial.
\end{conjecture}
This will be proved in section \ref{sec:secondbracket}, but first a brief exposition of the motivating arguments supporting the conjecture:
\begin{itemize}
	\item The symmetry between the two Poisson brackets: it should not matter which one is chosen to generate the equations of the hierarchy.
	\item Examples such as the KdV and Gelfand-Dickey hierarchies are known to have bi-Hamiltonian structures (see \cite{D03}).
	\item In \cite{DZ98}, it has been shown that the infinitesimal deformation of the second bracket, i.e., the $\epsilon^2$-term of its $\epsilon$-expansion, is polynomial.
\end{itemize}
\end{comment}

\subsection{Organization of the paper} In Section~\ref{sec:hamiltonian} we set the appropriate formalism to work with Hamiltonian and bi-Hamiltonian structures. In Section~\ref{sec:DZhierarchy} we recall the construction of the principal Dubrovin-Zhang hierarchy from a CohFT and endow it with a Hamiltonian structure of hydrodynamic type. We then deform it to the full hierarchy and show that it inherits a Hamiltonian structure. In Section~\ref{sec:secondbracket} we prove that in the case the underlying CohFT is conformal, the DZ hierarchy is also bi-Hamiltonian, with the second Hamiltonian structure having singularities of a very particular type (we reprove a result of Dubrovin and Zhang on that) and being uniquely determined by the bi-Hamiltonian recursion relation. In particular, we prove there that the constant term of the second bracket it polynomial. In Section~\ref{sec:LiuPand} we recall the Liu--Pandharipande relations in the tautological ring, and summarize the most important corollaries. In Section~\ref{sec:vanishing} we prove that all terms that must vanish for degree reasons once the conjecture of Dubrovin and Zhang on polynomiality of the second bracket holds actually do vanish.

\subsection{Acknowledgments} We thank A.~Buryak and G.~Carlet for useful discussions. We also thank the anonymous referees for useful remarks. The authors were supported by the Netherlands Organization for Scientific Research.  S. S. was also partially supported by International Laboratory of Cluster Geometry NRU HSE, RF Government grant, ag. ~\textnumero{} 75-15-2021-608 dated 08.06.2021.
 
 \section{Hamiltonian structures} \label{sec:hamiltonian} In this section, we explain the $\theta$-formalism first introduced in \cite{getzler} and further developed in \cite{LZ13} and \cite{DLZ17} to work with Hamiltonian and bi-Hamiltonian structures. The main addition to their theory is the completion of the differential polynomial algebra to allow certain singularities.
 
Let $M$ be a formal germ of an $N$-dimensional smooth manifold. We define a formal supermanifold $\hat{M}$ by describing its ring of functions. A system of local coordinates on $\hat{M}$ is given by $\{ u^1, \dots, u^N, \theta_1, \dots, \theta_N \}$, where $\{ u^1, \dots, u^N \}$ is a system of local coordinates of $M$ and $\{ \theta_1, \dots, \theta_N \}$ are the corresponding dual coordinates. Note the latter are Grassmann variables, i.~e., $\theta_\alpha \theta_\beta + \theta_\beta \theta_\alpha = 0$ for all $\alpha, \beta$.

Consider now the infinite jet space of $\hat{M}$, $J^\infty (\hat{M})$. A system of local coordinates in $J^\infty (\hat{M})$ is given by $\{ u^\alpha_p, \theta_\beta^q \}_{1 \leq \alpha, \beta \leq N}^{p,q \geq 0}$ and we identify $u^\alpha_0=u^\alpha$ and $\theta_\alpha^0 = \theta_\alpha$, $\alpha =1 ,\dots, N$. 
%From now on, we work on a local patch with those coordinates.

\begin{definition} The \emph{differential polynomial algebra} $\hat{\mathcal{A}}$ is defined as
	\begin{align}
	\hat{\mathcal{A}} = %\mathcal{C}^{\infty}(U) 
	\mathbb{C}[[ u^\alpha_0 = u^\alpha, u^\alpha_1, u^\alpha_2, \dots, \theta_\alpha^0 = \theta_\alpha, \theta_\alpha^1, \theta_\alpha^2, \dots \ | \ 1 \leq \alpha \leq N]].
	\end{align}
%	where $U$ is the corresponding local patch of $M$ with local coordinates $\{ u^\alpha \}$.
The term differential polynomial in $u$ (respectively, in $u,\theta$) means for us a polynomial in $u^\alpha_s$, $s\geq 1$, (respectively, in $u^\alpha_s$, $s\geq 1$, and $\theta_\alpha^s$, $s\geq 0$) with formal power series in $u^\alpha$ as coefficients.
\end{definition}

The differential polynomial algebra admits two gradations: the standard gradation 
\begin{align}
\deg u^\alpha_s = \deg \theta_\alpha^s = s, \qquad 
s\geq 0.
%\deg f = 0, \ f \in \mathcal{C}^\infty (U)
\end{align}
and the super gradation
\begin{align}
\deg \theta_\alpha^s = 1, \qquad \deg u^\alpha_s = 0, \qquad s\geq 0.
\end{align}
Let $\hat{\mathcal{A}}^p$ and $\hat{\mathcal{A}}_p$  be the degree $p$ components of $\hat{\mathcal{A}}$ with respect to the super and standard gradations, respectively. Let $\hat{\mathcal{A}}^p_d = \hat{\mathcal{A}}^p \cap \hat{\mathcal{A}}_d$. Note that $\mathcal{A} = \hat{\mathcal{A}}^0$ is the differential polynomial algebra on $M$. In the following, Einstein's summation convention applies to Greek indices, but not to Latin ones. The space $\hat{\mathcal{A}}$ is endowed with a total derivative
\begin{align}
\partial_x = \sum_{p \geq 0} \left( u^\alpha_{p+1} \frac{\partial}{\partial u^\alpha_p} + \theta_\alpha^{p+1} \frac{\partial }{\partial \theta_\alpha^p} \right),
\end{align}
which preserves the super gradation on $\hat{\mathcal{A}}$ and increases the standard gradation by $1$.

\begin{definition} The space of polyvector fields $\hat{\mathcal{F}}$ is defined as the quotient space of $\hat{\mathcal{A}}$ by $\partial_x \hat{\mathcal{A}}$ and the constant functions. 
\end{definition}
The projection map is denoted by $\int: \hat{\mathcal{A}} \rightarrow \hat{\mathcal{F}}$. Since $\partial_x$ is homogeneous with respect to the gradations, the subspaces $\hat{\mathcal{F}}^p$, $\hat{\mathcal{F}}_d$ and $\hat{\mathcal{F}}^p_d$ are well defined. The elements of $\hat{\mathcal{F}}^p$ are called $p$-vectors. It is possible to write a $p$-vector as a sum of its homogeneous components in the standard degree: for $f \in \hat{\mathcal{F}}^p$, we write
\begin{align}
f = f_0 + f_1 + f_2 + \dots, \qquad f_k \in \hat{\mathcal{F}}^p_k.
\end{align}
It will be useful later to introduce a formal parameter $\epsilon$ 
(the dispersion parameter) 
to keep track of the degree by rescaling $x\mapsto \epsilon x$. That is, if $s = \min\{ k| f_k \not= 0 \}$, then we write
\begin{align}
f = f_s +  \epsilon f_{s+1} + \epsilon^2 f_{s+2} + \dots, \qquad f_k \in \hat{\mathcal{F}}^p_k
\end{align}

Having defined the functionals, the goal is to define a graded Lie bracket on $J^\infty (\hat{M})$ that extends the usual Schouten bracket on $\hat{M}$, given by
\begin{align}
[P, Q] = \frac{\partial P}{\partial \theta_\alpha} \frac{\partial Q}{\partial u^\alpha} + (-1)^p \frac{\partial P}{\partial u^\alpha} \frac{\partial Q}{\partial \theta_\alpha}
\end{align}
for $P \in \hat{\mathcal{A}}^p_0$ and $Q \in \hat{\mathcal{A}}^q_0$. For this purpose, define the variational derivatives of $f \in \hat{\mathcal{A}}$:
\begin{align}
\frac{\delta f}{\delta \theta_\gamma} = \sum_{p \geq 0} (-\partial_x)^p \frac{\partial f}{\partial \theta_\gamma^p}, \qquad \frac{\delta f}{\delta u^\gamma} = \sum_{p \geq 0} (-\partial_x)^p \frac{\partial f}{\partial u^\gamma_p}.
\end{align}
Since $\frac{\delta}{\delta \theta_\gamma} \circ \partial_x = \frac{\delta}{\delta u^\gamma} \circ \partial_x = 0$, the operators above can be defined on the space of functionals $\hat{\mathcal{F}}$. Now we can define the bracket
\begin{align}
[\cdot, \cdot]: \hat{\mathcal{F}}^p \times \hat{\mathcal{F}}^q &\longrightarrow \hat{\mathcal{F}}^{p+q-1} \\ 
(P,Q) &\longmapsto [P,Q] = \int \left( \frac{\delta P}{\delta \theta_\alpha} \frac{\delta Q}{\delta u^\alpha} + (-1)^p \frac{\delta P}{\delta u^\alpha} \frac{\delta Q}{\delta \theta_\alpha} \right) dx,
\end{align}
which we call the Schouten bracket on $J^\infty(\hat{M})$. The next theorem shows that this bracket gives $\hat{\mathcal{F}}$ a graded Lie algebra structure:

\begin{theorem}[\cite{LZ11}]
	Let $P \in \hat{\mathcal{F}}^p, Q \in \hat{\mathcal{F}}^q, R \in \hat{\mathcal{F}}^r$. Then
	\begin{itemize}
		\item $[P,Q] = (-1)^{pq} [Q, P]$.
		\item $(-1)^{pr} [[P,Q], R] + (-1)^{qp} [[Q,R],P] + (-1)^{rq} [[R,P],Q] = 0$.
	\end{itemize}
\end{theorem}

Finally, we are ready to define Hamiltonian and bi-Hamiltonian structures.
\begin{definition}
	A Poisson bivector or Hamiltonian structure is a bivector $P \in \hat{\mathcal{F}}^2$ such that $[P,P] = 0$. A bi-Hamiltonian structure is a pair $(P_1, P_2) $ of Poisson bivectors satisfying $[P_1, P_2] = 0$. 
\end{definition}

\begin{remark}[Poisson bivectors, Poisson brackets, Poisson operators] \label{remark:poisson}
	A Poisson bivector $P$ defines a Poisson bracket on $%\mathcal{F} =
	 \hat{\mathcal{F}}^0$
	\begin{align}
	\{ \cdot, \cdot \}_P\colon \hat{\mathcal{F}}^0\times \hat{\mathcal{F}}^0&\longrightarrow \hat{\mathcal{F}}^0 \\
	(F,G) &\longmapsto \{ F,G \}_P = [[P,F], G].
	\end{align}
	On the other hand, given a Poisson bracket $\{ \cdot, \cdot \}_P$, there exist unique $P^{\alpha \beta}_s \in \mathcal{A}$ satisfying $\sum_{s \geq 0} P^{\alpha \beta}_s \partial_x^s = \sum_{s \geq 0} (-1)^{s+1} \partial_x^s P^{\beta \alpha}_s$ such that the Poisson operator $P^{\alpha \beta} = \sum_{s \geq 0} P^{\alpha \beta}_s \partial_x^s$ gives the bracket
	\begin{align}
	\{ F, G \}_P = \int \frac{\delta F}{\delta u^\alpha} P^{\alpha \beta} \frac{\delta G}{\delta u^\beta} dx.
	\end{align}
	It is clear that for 
	\begin{align}
	P = \frac{1}{2} \int \theta_\alpha P^{\alpha \beta} (\theta_\beta) dx= \frac{1}{2} \sum_{s \geq 0} \int P^{\alpha \beta}_s \theta_\alpha \theta_\beta^s dx
	\end{align}
	both definitions of $\{ \cdot, \cdot \}_P$ coincide. Thus, we will use the terms Poisson bivector, Poisson operator and Poisson bracket interchangeably from this point.
\end{remark}

\begin{example}
	Particularly interesting examples of Poisson bivectors are those of \emph{hydrodynamic type}, i.e., those of the standard degree $1$. $P \in \hat{\mathcal{F}}^2_1$ is a Poisson bivector of hydrodynamic type if and only if it takes the form (see \cite{DN83})
	\begin{align}
	P = \frac{1}{2} \int ( g^{\alpha \beta} (u) \theta_\alpha \theta_\beta^1   + \Gamma^{\alpha \beta}_\gamma (u) u^\gamma_1 \theta_\alpha \theta_\beta  )dx,
	\end{align}
	where $g$ is a flat metric on $M$ and $\Gamma^{\alpha \beta}_\gamma$ are the contravariant Christoffel symbols of its Levi-Civita connection. 
\end{example}

The Poisson bivectors considered in this paper will be deformations in even degrees of Poisson bivectors of hydrodynamic type, i.e., $P \in \hat{\mathcal{F}}^2$ such that $P|_{\epsilon = 0}$ is a Poisson bracket of hydrodynamic type. For any such $P$, we can find unique $P^{\alpha \beta}_{g,s} \in \mathcal{A}_{2g+1-s}$ satisfying $\sum_{s = 0}^{2g+1} P^{\alpha \beta}_{g,s} \partial_x^s = \sum_{s = 0}^{2g+1} (-1)^{s+1} \partial_x^s P^{\beta \alpha}_{g,s}$ such that $P$ can be written as
\begin{align}
P = \frac{1}{2} \int \sum_{g = 0}^{\infty} \epsilon^{2g} \sum_{s=0}^{2g+1} P^{\alpha \beta}_{g,s} \theta_\alpha \theta_\beta^s dx.
\end{align}

\subsection{Changes of coordinates} We want the Schouten bracket to remain invariant under changes of coordinates. First, we introduce the groups of transformations that will be considered.
\begin{definition}
A \emph{Miura transformation} is a formal change of coordinates $u^\alpha \rightarrow \tilde{u}^\alpha$ of the form
		\begin{align}
		\tilde{u}^\alpha = u^\alpha + \sum_{k=1}^{\infty} \epsilon^k F^\alpha_k (u; u_1, \dots, u_k) \label{eqn:Miura}
		\end{align}
		where $F^\alpha_k \in \mathcal{A}_k$. 
\end{definition}
In the literature, Miura transformations of the form \eqref{eqn:Miura} with an arbitrary diffeomorphism $F^\alpha_0(u)$ as the leading term are often considered. However, all (quasi-)Miura transformations studied in this paper have $F^\alpha_0(u) = u^\alpha$, thus justifying our more restricted definition, usually known as \emph{Miura transformations close to the identity}. It is important to study the behavior of the Schouten bracket under Miura transformations:
\begin{proposition}[\cite{LZ11}] \label{prop:changeschouten}
	A Miura transformation \eqref{eqn:Miura} $u^\alpha \rightarrow \tilde{u}^\alpha$ induces a change of variables
	\begin{align}
	\theta_\beta = \sum_{s \geq 0} (- \partial_x)^s \left( \frac{\partial \tilde{u}^\alpha}{\partial u_s^\beta} \tilde{\theta}_\alpha \right) \label{eqn:changetheta}
	\end{align}
	such that the Schouten bracket remains invariant.
\end{proposition}
\begin{proof}
	Let us sketch the proof with the help of the formulas given in \cite{LZ11}. The change for the variational derivatives is given by
	\begin{align}
	\frac{\delta}{\delta \tilde{u}^\alpha} = \sum_{p \geq 0} (-\partial_x)^p \circ \frac{\partial u^\beta}{\partial \tilde{u}^\alpha_p} \circ \frac{\delta}{\delta u^\beta}, \label{eqn:fracdeltau}
	\end{align}
	 Analogously,
	\begin{align}
	\frac{\delta }{\delta \tilde{\theta}_\alpha} = \sum_{t \geq 0} (-\partial_x)^t \circ \frac{\partial \theta_\beta}{\partial \tilde{\theta}_{\alpha}^t} \circ \frac{\delta}{\delta \theta_\beta}.
	\end{align}
	From \eqref{eqn:changetheta}, we get
	\begin{align}
	\frac{\partial \theta_\beta}{\partial \tilde{\theta}_\alpha^t} = (-1)^t \sum_{s \geq 0} \begin{pmatrix}
	s + t \\ t
	\end{pmatrix} (- \partial_x)^s \left( \frac{\partial \tilde{u}^\alpha}{\partial u^\beta_{s+t}} \right),
	\end{align}
	so
	\begin{align}
	\frac{\delta }{\delta \tilde{\theta}_\alpha} = \sum_{t \geq 0} \partial_x^t \circ \sum_{s \geq 0} \begin{pmatrix}
	s + t \\ t
	\end{pmatrix} (- \partial_x)^s \left( \frac{\partial \tilde{u}^\alpha}{\partial u^\beta_{s+t}} \right) \circ \frac{\delta}{\delta \theta_\beta}. \label{eqn:fracdeltatilde}
	\end{align}
	Let $P, Q \in \hat{\mathcal{F}}$. The result follows after replacing \eqref{eqn:fracdeltau} and \eqref{eqn:fracdeltatilde} in the expression
	\begin{align}
	[P, Q] = \int \frac{\delta P}{\delta \tilde{u}^\alpha} \frac{\delta Q}{\delta \tilde{\theta}_\alpha} dx
	\end{align}
\end{proof}

\begin{example}[Transformation rule for Poisson brackets] \label{example:transformationrule}
	Let $P = \frac{1}{2} \int \theta_\alpha P^{\alpha \beta} (\theta_\beta)$ be a Poisson bivector, and consider the change of variables \eqref{eqn:Miura}. In the new coordinates, $P$ takes the form
	\begin{align}
	P & = \frac{1}{2} \int \sum_{s \geq 0} (- \partial_x)^s \left( \frac{\partial \tilde{u}^\gamma}{\partial u_s^\alpha} \tilde{\theta}_\gamma \right) P^{\alpha \beta} \left( \sum_{t \geq 0} (- \partial_x)^t \left( \frac{\partial \tilde{u}^\sigma}{\partial u_s^\beta} \tilde{\theta}_\sigma \right) \right) dx \\
	\notag 
	& = \frac{1}{2} \int \sum_{s, t \geq 0}  \tilde{\theta}_\gamma \frac{\partial \tilde{u}^\gamma}{\partial u_s^\alpha} \partial_x^s \circ   P^{\alpha \beta} \circ (- \partial_x)^t \circ  \frac{\partial \tilde{u}^\sigma}{\partial u_s^\beta} (\tilde{\theta}_\sigma) dx.
	\end{align}
	Therefore, the transformed Poisson operator $\tilde P^{\gamma\sigma}$ is given by
	\begin{align}
		\tilde P^{\gamma\sigma} = 
	\sum_{s, t \geq 0}  \frac{\partial \tilde{u}^\gamma}{\partial u_s^\alpha} \partial_x^s \circ   P^{\alpha \beta} \circ (- \partial_x)^t \circ  \frac{\partial \tilde{u}^\sigma}{\partial u_s^\beta}.
	\end{align}
	
\end{example}

\subsection{Differential rational functions} For applications in enumerative geometry the spaces $\hat{\mathcal{A}}$ and $\hat{\mathcal{F}}$ are too restrictive. We will complete them by allowing certain singularities.
\begin{definition}
	A differential rational function of type $(1,1)$ is a function $f$ of the form
	\begin{align}
	f = f_s + \epsilon f_{s+1} + \epsilon^2 f_{s+2} + \dots, 
	\end{align}
	where the functions $f_k$ only depend on finitely many derivatives $u, u_1, \dots, u_r$ and take the form
	\begin{align}
	f_k = (u^1_1)^k \sum_{l = 0}^{\infty} \frac{P_{k,l}}{(u^1_1)^l} \label{eqn:define11}
	\end{align}
	with $P_{k,l} = P_{k, l}(u_\bullet^\bullet, \theta_{\bullet}^\bullet)$ is a homogeneous differential polynomial of standard degree $l$ that does not depend on $u^1_1$. The space of differential rational functions of type $(1,1)$ is denoted by $\hat{\mathcal{B}}$. The space of rational polyvector fields of type $(1,1)$ is defined as the quotient of $\hat{\mathcal{B}}$ by $\partial_x \hat{\mathcal{B}}$ and the constant functions and denoted by $\hat{\mathcal{Q}}$. 
\end{definition}
The next step is to enlarge the Miura group to allow transformations given by differential rational functions.
\begin{definition} A \emph{quasi-Miura transformation} is a change of variables of the form
	\begin{align}
	\tilde{u}^\alpha = u^\alpha + \sum_{k=1}^{\infty} \epsilon^k F^\alpha_k (u; u_1, \dots, u_{n_k}) \label{eqn:QM}.
	\end{align}
	where the functions $F^\alpha_k$ are homogeneous rational functions in the derivatives $u_1, \dots, u_{n_k}$ of standard degree $k$. A quasi-Miura transformation is of $(1,1)$-type if it takes the form
\begin{align}
\tilde{u}^\alpha = u^\alpha + \sum_{k=1}^{\infty} (\epsilon u^1_1)^k \sum_{l=0}^{\infty} \frac{F^\alpha_{k,l} (u; u_1, \dots, u_{n_k})}{(u^1_1)^l}, \label{eqn:qMbeta}
\end{align}
where $F^\alpha_{k,l}$ is a differential polynomial of degree $l$ with $\frac{\partial F^\alpha_{k,l}}{ \partial u^1_1} = 0$. 
\end{definition}

\begin{remark}
	\begin{itemize}
		\item The key aspect of the space $\hat{\mathcal{Q}}$ is that all notions introduced before for $\hat{\mathcal{F}}$ are still well-defined for $\hat{\mathcal{Q}}$: variational derivatives, the Schouten bracket and both gradations. It is also possible to define Hamiltonian and bi-Hamiltonian structures of type $(1,1)$ in the same way it was done for their polynomial analogues. 
		\item Note the proof of Proposition \ref{prop:changeschouten} does not use polynomiality at any point. The only requirement is that the variational derivatives still make sense for the transformed polyvector fields, so it still holds if we consider quasi-Miura transformations of type $(1,1)$ instead of Miura transformations.
		\item Applying \eqref{eqn:qMbeta} to an element of $\hat{\mathcal{F}}^r$ yields an element of $\hat{\mathcal{Q}}^r$. The key question is whether this new element is polynomial in the variables $\tilde{u}$.
	\end{itemize}
\end{remark}

\section{Dubrovin--Zhang hierarchy} \label{sec:DZhierarchy} In this section we recall the construction of the DZ hierarchy as done in \cite{DZ01, BPS122, BPS12} as well as some of its most important properties. For the full details, we refer the reader to those articles.

\subsection{Cohomological field theories} \label{subsec:cohft} In this section, we briefly recall the definition and the most important properties of a cohomological field theory, firstly introduced by Kontsevich and Manin in \cite{KM94}. Let $V$ be an $N$-dimensional vector space over $\mathbb{C}$ equipped with a scalar product $(\cdot, \cdot)$. Choose a basis $\{ e_1, e_2, \dots, e_N \}$ of $V$ and let $\eta_{\alpha \beta} = (e_\alpha, e_\beta)$. In the rest of the paper we use $\eta_{\alpha \beta}$ to lower indices and its inverse $\eta^{\alpha \beta}$ to raise them.
\begin{definition}[CohFT]
	A \emph{cohomological field theory} with unit $e_1$ is a collection of linear homomorphisms
	\begin{align}
	c_{g,n}: V^{\otimes n} \rightarrow H^{2*} \left( \Mgn; \mathbb{C} \right), \qquad 2g-2+n > 0
	\end{align}
	such that
	\begin{itemize}
		\item $c_{g,n}$ is $S_n$-equivariant, where $S_n$ acts on $V^{\otimes n}$ by permutation of the factors and on $H^{2*} \left( \Mgn; \mathbb{C} \right)$ by permutation of the marked points.
		\item For any gluing map $\rho: \overline{\mathcal{M}}_{g_1, n_1+1} \times \overline{\mathcal{M}}_{g_2, n_2+1} \rightarrow \overline{\mathcal{M}}_{g_1+g_2, n_1+ n_2}$, 
		\begin{align}
		& \rho^* c_{g_1+g_2, n_1 + n_2} (v_1 \otimes \dots \otimes v_{n_1+n_2}) =
		\\ \notag & c_{g_1, n_1+1}(v_1 \otimes \dots \otimes v_{n_1} \otimes e_\alpha) \eta^{\alpha \beta} c_{g_2, n_2 + 1} (v_{n_1+1} \otimes \dots \otimes v_{n_2} \otimes e_\beta).
		\end{align}
		\item For any gluing map $\sigma: \overline{\mathcal{M}}_{g-1, n+2} \rightarrow \Mgn$, 
		\begin{align}
		\sigma^{*} c_{g,n} (v_1 \otimes \dots \otimes v_n) = c_{g-1, n+2} (v_1 \otimes \dots \otimes v_n \otimes e_\alpha \otimes e_\beta) \eta^{\alpha \beta}.
		\end{align}
		\item For any forgetful map $\pi: \overline{\mathcal{M}}_{g, n+1} \rightarrow \Mgn$,
		\begin{align}
		\pi^* c_{g,n}(v_1 \otimes \dots \otimes v_n) = c_{g, n+1} (v_1 \otimes \dots \otimes v_n \otimes e_1)
		\end{align}
		\item $c_{0,3} (v_1 \otimes v_2 \otimes e_1) = (v_1, v_2)$
	\end{itemize}
\end{definition}

We associate to a CohFT $\{c_{g,n}\}$ a formal power series in the variables $\{ t^\alpha_d \}^{1 \leq \alpha \leq N}_{d \geq 0}$. Firstly, define the correlation functions as
\begin{align}
\langle \prod_{i=1}^{n} \tau_{d_i} (v_i) \rangle_g \coloneqq \int_{\Mgn} c_{g,n}( \otimes_{i=1}^n v_i  ) \prod_{i=1}^{n} \psi_i^{d_i},
\end{align}
where $\psi_i$ is the psi-class at the $i$-th marked point. Define the partition function and the potential as
\begin{align}
\tau \coloneqq \exp(\epsilon^{-2} F), \qquad F \coloneqq \sum_{g \geq 0} \epsilon^{2g} F_g, \qquad
F_g \coloneqq \sum_{\substack{n \geq 0 \\ 2g-2+n > 0}} \frac{1}{n!} \sum_{\substack{1 \leq \alpha_1, \dots, \alpha_n \leq N \\ d_1, \dots, d_n \geq 0}} \langle \prod_{i = 1}^{n} \tau_{d_i}(e_{\alpha_i}) \rangle_g \prod_{i=1}^{n} t_{d_i}^{\alpha_i}.
\end{align}
The following identities satisfied by the potential will be used in the text (see e.~g.~\cite{wittenAlgGeom}):
\begin{itemize}
	\item String equation
	\begin{align} \label{eqn:string}
	%\langle \tau_0 (e_1) \prod_{i=1}^{n} \tau_{d_i} (e_{\alpha_i}) \rangle_g = \sum_{j=1}^{n} \langle \tau_{d_j-1}(e_{\alpha_j}) \prod_{i \not= j} \tau_{d_i}(e_{\alpha_i}) \rangle_g,\\
	\frac{\partial F}{\partial t_0^1} = \sum_{p \geq 0} t_{p+1}^\alpha \frac{\partial F}{\partial t^{\alpha}_p} + \frac{1}{2} \eta_{\alpha \beta} t_0^\alpha t_0^\beta + \epsilon^2 \langle \tau_0 (e_1) \rangle_1.
	\end{align}
	\item Dilaton equation
	\begin{align} \label{eqn:dilaton}
	%\langle \tau_1(e_1) \prod_{i = 1}^{n} \tau_{d_i} (e_{\alpha_i}) \rangle_g = (2g -2 + n)\langle \prod_{i = 1}^{n} \tau_{d_i} (e_{\alpha_i}) \rangle_g, \\
	\frac{\partial F}{\partial t_1^1} = \epsilon \frac{\partial F}{\partial \epsilon} + \sum_{p \geq 0} t^\alpha_p \frac{\partial F}{\partial t^\alpha_p} - 2 F + \epsilon^2 \frac{N}{24}.
	\end{align}
	\item WDVV (associativity) equations
	\begin{align}
	\frac{\partial^3 F_0}{\partial t^{\alpha_1}_{d_1} \partial t^{\alpha_2}_{d_2} \partial t_0^\alpha } \eta^{\alpha \beta} \frac{\partial^3 F_0}{\partial t^{\alpha_3}_{d_3} \partial t^{\alpha_4}_{d_4} \partial t_0^\beta } = \frac{\partial^3 F_0}{\partial t^{\alpha_2}_{d_2} \partial t^{\alpha_3}_{d_3} \partial t_0^\alpha } \eta^{\alpha \beta} \frac{\partial^3 F_0}{\partial t^{\alpha_1}_{d_1} \partial t^{\alpha_4}_{d_4} \partial t_0^\beta }. \label{eqn:wddv}
	\end{align}
	\item Topological recursion relation in genus 0 (TRR-0):
	\begin{align} 
	\frac{\partial^3 F_0}{\partial t^{\alpha_1}_{d_1+1} \partial t^{\alpha_2}_{d_2} \partial t^{\alpha_3}_{d_3}} = \frac{\partial^2 F_0}{\partial t^{\alpha_1}_{d_1} \partial t^{\alpha}_{0} } \eta^{\alpha \beta} \frac{\partial^3 F_0}{\partial t^{\beta}_{0} \partial t^{\alpha_2}_{d_2} \partial t^{\alpha_3}_{d_3}} \label{eqn:TRR0}
	\end{align}
\item The Liu--Pandharipande relations. We recall them in Section~\ref{sec:LiuPand}.
\end{itemize} 
There are many other universal differential equations for $F_g$, $g\geq 0$. Basically, any relation in the tautological ring of the moduli space of curves implies such an equation.

\subsection{The principal hierarchy} \label{subsec:prinhierar} From now on, all CohFTs considered in this article are semi-simple, i.~e., the functions
\begin{align}
c^\lambda_{\alpha \beta} (t) = \eta^{\lambda \gamma} \frac{\partial^3 F^{\textrm{Frob}}}{\partial t^\gamma \partial t^\alpha \partial t^\beta} ,
\end{align}
where $F^\textrm{Frob} = F_0|_{t^\bullet_{\geq 1} = 0}$ give the structure constants of a semi-simple associative algebra at $t=0$.

Let $t^1_0\mapsto t^{1}_0+x$. Consider the \emph{two-point correlators in genus $g$}
\begin{align}
\Omega^{[g]}_{\alpha, p; \beta, q} = \frac{\partial^2 F_g}{\partial t^\alpha_p \partial t^\beta_q}
\end{align}
and the variables
\begin{align}
v_\beta = \frac{\partial^2 F_0}{\partial t^\beta_0 \partial t^1_0}, \qquad v^\alpha = \eta^{\alpha \beta} v_\beta, \qquad v^\alpha_k = \partial_x^k  v^\alpha.
\end{align}
\begin{proposition}[\cite{Dubrovin1996,BPS12}]
	The two point correlators in genus 0 are given by
	\begin{align}
	\Omega^{[0]}_{\alpha, p; \beta, q} (t_0, t_1, t_2,  \dots) = \Omega^{[0]}_{\alpha, p; \beta, q} (v, 0, 0, \dots)
	\end{align}
\end{proposition}
As a consequence of the tau-symmetry of $\Omega^{[g]}_{\alpha, p; \beta, q}$, i.e., the expression \begin{align}
\frac{\partial}{\partial t^\gamma_r} \Omega^{[g]}_{\alpha, p; \beta, q}
\end{align}
being invariant under any permutation of $(\alpha, p) \leftrightarrow (\beta,q) \leftrightarrow (\gamma, r)$, the variables $v^\alpha$ satisfy the system of equations
\begin{align}
\frac{\partial v^\alpha}{\partial t^\beta_q} = \eta^{\alpha \gamma} \partial_x (\Omega^{[0]}_{\gamma, 0; \beta, q}). \label{eqn:taucover}
\end{align}
%where $\partial_x = \frac{\partial}{\partial t^1_0}$. 
The goal is to rewrite the equations \eqref{eqn:taucover} in Hamiltonian form. First, define the Hamiltonian densities
\begin{align}
h_{\alpha, p}(v) \coloneqq \Omega^{[0]}_{\alpha, p+1; 1, 0}. \label{eqn:defhamil}
\end{align}
Consider the Poisson operator of hydrodynamic type $P^{\alpha \beta} = \eta^{\alpha \beta} \partial_x$ or, equivalently by Remark~\ref{remark:poisson}, the Poisson bivector 
\begin{align}
P = \frac{1}{2} \int \theta_\alpha \eta^{\alpha \beta} \theta^1_\beta dx\in \hat{\mathcal{F}}^2_1. \label{eqn:firstbracket}
\end{align}
%We have
\begin{proposition}[\cite{Dubrovin1996,BPS12}] The Hamiltonians $\bar{h}_{\alpha, p} = \int h_{\alpha, p}dx$ satisfy:
	\begin{align}
 \frac{\delta \bar{h}_{\alpha, p}}{\delta v^\gamma} \eta^{\gamma\sigma} \partial_x \frac{\delta \bar{h}_{\beta, q}}{\delta v^\sigma} = \partial_x \Omega^{[0]}_{\alpha, p+1; \beta, q}.
	\end{align}
	In particular, they Poisson-commute $\{ \bar{h}_{\alpha, p}, \bar{h}_{\beta, q} \}_P = 0$ and the system \eqref{eqn:taucover} can be rewritten as a Hamiltonian system, called the principal or dispersionless Dubrovin--Zhang  hierarchy:
	\begin{align}
	\frac{\partial v^\alpha}{\partial t^\beta_q} = \eta^{\alpha \gamma} \partial_x \frac{\delta \bar{h}_{\beta, q}}{\delta v^\gamma} \label{eq:principalhierarchy}
	\end{align}
\end{proposition}

\subsection{The full hierarchy} \label{sec:fullhierarchy}

The principal hierarchy constructed above ``forgets" the information of the CohFT carried by $F_g$ for $g \geq 1$, in other words, no information is lost if we set $\epsilon = 0$ at the beginning. Here we construct the full hierarchy.
Consider the variables $w$ and the two point correlators:
\begin{align}
w^\alpha = \eta^{\alpha \beta} \frac{\partial^2 F}{\partial t^\beta_0 \partial t^1_0}, \qquad \Omega_{\alpha, p; \beta, q} = \frac{\partial^2 F}{\partial t^\alpha_p \partial t^\beta_q} = \sum_{g=0}^{\infty} \epsilon^{2g} \Omega^{[g]}_{\alpha, p; \beta, q}
\end{align}
As a consequence of the tau-symmetry of $\Omega_{\alpha, p; \beta, q}$, i.e., the expression 
\begin{align}
\frac{\partial}{\partial t^\gamma_r} \Omega_{\alpha, p; \beta, q}
\end{align}
being invariant under any permutation of $(\alpha, p) \leftrightarrow (\beta,q) \leftrightarrow (\gamma, r)$, the variables $w^\alpha$ satisfy the system of equations
\begin{align}
\frac{\partial w^\alpha}{\partial t^\beta_q} = \eta^{\alpha \gamma} \partial_x (\Omega_{\gamma, 0; \beta, q}), \label{eqn:taucoverfull}
\end{align}
The goal is to endow the system \eqref{eqn:taucoverfull} with a Hamiltonian structure as it was done for the principal hierarchy in Section \ref{subsec:prinhierar}. For this, we recall an important result, known as the $(3g-2)$-property:
\begin{proposition}[see e.~g. \cite{BPS12}] \label{prop:3g-2}
	For $g \geq 1$, there exist functions $P_0^{[g]}, \dots, P_{3g-2}^{[g]}$ such that
	\begin{align}
	F_g (t_0, t_1, \dots) = F_g (P_0^{[g]}(v, v_1, \dots, v_{3g-2}), \dots, P_{3g-2}^{[g]} (v, v_1, \dots, v_{3g-2}), 0, 0, \dots  )
	\end{align}
\end{proposition}
As a consequence of the $(3g-2)$-property, $\Omega^{[g]}$ only depends on $v, v_1, \dots, v_{3g}$, so the expansion of $w$ in terms of $v$ takes the form
\begin{align}
w^\alpha = v^\alpha + \eta^{\alpha \beta} \sum_{g=1}^{\infty} \epsilon^{2g} \Omega^{[g]}_{\beta, 0; 1, 0}(v, v_1, \dots, v_{3g}) \label{eqn:quasimiuraw}.
\end{align}
The next proposition shows that the string and dilaton equations give \eqref{eqn:quasimiuraw} the required regularity to use the framework developed in Section \ref{sec:hamiltonian}.
\begin{proposition}[see~\cite{BDGR-20}] \label{prop:type11}
	The transformation $v^\alpha \rightarrow w^\alpha$ is a quasi-Miura transformation of $(1,1)$-type.
\end{proposition}
\begin{proof}
	Dilaton equation \eqref{eqn:dilaton} implies the following relation between the coefficients of $F_g$
	\begin{align}
	[\prod_{i=1}^{n} t_{d_i}^{\alpha_i}] F_g = \frac{(2g-3+n)! d!}{(2g-3+n+d)!} [\prod_{i=1}^{n} t_{d_i}^{\alpha_i} (t_1^1)^d] F_g.
	\end{align}
	This allows us to rewrite $F_g$ ($g \geq 2$) in the following form:
	\begin{align}
	F_g &= \sum_{\substack{n \geq 0 }} \frac{1}{n!} \sum_{\substack{1 \leq \alpha_1, \dots, \alpha_n \leq N \\ d_1, d_2, \dots, d_n \geq 0 \\ (\alpha_k, d_k) \not= (1,1)}} \langle \prod_{i=1}^{n} \tau_{d_i}(e_{\alpha_i}) \rangle_g \prod_{i=1}^{n} t_{d_i}^{\alpha_i} \sum_{d \geq 0} \frac{(2g - 3 +n + d)!}{(2g -3 + n)! d!} (t^1_1)^d \\ \notag 
	&= \sum_{\substack{n \geq 0}} \frac{1}{n!} \sum_{\substack{1 \leq \alpha_1, \dots, \alpha_n \leq N \\ d_1, d_2, \dots, d_n \geq 0 \\ (\alpha_k, d_k) \not= (1,1)}} \langle \prod_{i=1}^{n} \tau_{d_i}(e_{\alpha_i}) \rangle_g \prod_{i=1}^{n} t_{d_i}^{\alpha_i} \left( \frac{1}{1-t^1_1} \right)^{2g-2+n} \\ \notag 
	&= \left( \frac{1}{1-t^1_1} \right)^{2g-2} \sum_{\substack{n \geq 0}} \frac{1}{n!} \sum_{\substack{1 \leq \alpha_1, \dots, \alpha_n \leq N \\ d_1, d_2, \dots, d_n \geq 0 \\ (\alpha_k, d_k) \not= (1,1)}} \langle \prod_{i=1}^{n} \tau_{d_i}(e_{\alpha_i}) \rangle_g \prod_{i=1}^{n} \left(\frac{t_{d_i}^{\alpha_i}}{1-t^1_1} \right) .
	\end{align}
	For $F_0$ and $F_1$ we have to mind the unstable correlation functions:
	\begin{align}
	F_1 &= \sum_{\substack{n \geq 1}} \frac{1}{n!} \sum_{\substack{1 \leq \alpha_1, \dots, \alpha_n \leq N \\ d_1, d_2, \dots, d_n \geq 0 \\ (\alpha_k, d_k) \not= (1,1)}} \langle \prod_{i=1}^{n} \tau_{d_i}(e_{\alpha_i}) \rangle_1 \prod_{i=1}^{n} \left(\frac{t_{d_i}^{\alpha_i}}{1-t^1_1} \right) + \sum_{d \geq 0} \frac{1}{(d+1)!} \langle (\tau_1(e_1) )^{d+1} \rangle_1 (t^1_1)^{d+1} \\ \notag 
	&= \sum_{\substack{n \geq 1}} \frac{1}{n!} \sum_{\substack{1 \leq \alpha_1, \dots, \alpha_n \leq N \\ d_1, d_2, \dots, d_n \geq 0 \\ (\alpha_k, d_k) \not= (1,1)}} \langle \prod_{i=1}^{n} \tau_{d_i}(e_{\alpha_i}) \rangle_1 \prod_{i=1}^{n} \left(\frac{t_{d_i}^{\alpha_i}}{1-t^1_1} \right) + \langle \tau_1(e_1) \rangle_1 \sum_{d \geq 0} \frac{1}{(d+1)} (t^1_1)^{d+1} \\ \notag 
	&= \sum_{\substack{n \geq 1}} \frac{1}{n!} \sum_{\substack{1 \leq \alpha_1, \dots, \alpha_n \leq N \\ d_1, d_2, \dots, d_n \geq 0 \\ (\alpha_k, d_k) \not= (1,1)}} \langle \prod_{i=1}^{n} \tau_{d_i}(e_{\alpha_i}) \rangle_1 \prod_{i=1}^{n} \left(\frac{t_{d_i}^{\alpha_i}}{1-t^1_1} \right) + \frac{N}{24} \log\left( \frac{1}{1-t^1_1} \right) ;\\
	F_0 &=  \left(1-t^1_1\right)^{2} \sum_{\substack{n \geq 3}} \frac{1}{n!} \sum_{\substack{1 \leq \alpha_1, \dots, \alpha_n \leq N \\ d_1, d_2, \dots, d_n \geq 0 \\ (\alpha_k, d_k) \not= (1,1)}} \langle \prod_{i=1}^{n} \tau_{d_i}(e_{\alpha_i}) \rangle_0 \prod_{i=1}^{n} \left(\frac{t_{d_i}^{\alpha_i}}{1-t^1_1} \right) ,
	\end{align}
	where the last formula follows from the fact that $\langle \tau_{d_1}(e_{\alpha_1} ) \tau_{d_2} (e_{\alpha_2}) \tau_1(e_1) \rangle_0 = 0$ for all $\alpha_1, \alpha_2, d_1, d_2$.
	By the string equation \eqref{eqn:string} we have
	\begin{align}
	v^\alpha & = \eta^{\alpha \beta} \frac{\partial^2 F_0}{\partial t^\beta_0 \partial t_0^1} = t_0^\alpha + \sum_{p \geq 0} \eta^{\alpha \beta} t_{p+1}^\gamma \frac{\partial^2 F_0}{\partial t^\gamma_p \partial t_0^\beta} = t_0^\alpha + \mathcal{O}(t^2) \label{eqn:vstring} \\
	v^1_1 & = 1 + t_1^1 + \sum_{p \geq 0} t_{1}^\gamma t_{p+1}^\mu \eta^{1 \beta} \frac{\partial^3 F_0}{\partial t_0^\gamma \partial t_0^\beta \partial t_p^\mu} +  \sum_{p \geq 1}  t_{p+1}^\gamma \eta^{1 \beta} \frac{\partial^3 F_0}{\partial t_0^1 \partial t_0^\beta \partial t_p^\gamma} = 1 + t^1_1 + \mathcal{O}(t^2) \label{eqn:dependencev1}
	\end{align}
	To get the dependence of $v^1_1$ on $t_1^1$ we compute
	\begin{align}
	& v_1^1 = \eta^{1 \mu} \frac{\partial^3 F_0}{\partial t_0^\mu \partial t_0^1 \partial t_0^1} = \frac{1}{1 - t_1^1} \sum_{n \geq 0} \frac{1}{n!} \sum_{\substack{1 \leq \alpha_1, \dots, \alpha_n \leq N \\ d_1, d_2, \dots, d_n \geq 0 \\ (\alpha_k, d_k) \not= (1,1)}} \eta^{1 \mu} \langle \tau_0 (e_\mu) \tau_0 (e_1)^2 \prod_{i=1}^{n} \tau_{d_i} (e_{\alpha_i}) \rangle_0 \prod_{i=1}^{n} \left( \frac{t_{d_i}^{\alpha_i}}{1-t_1^1} \right);  \\
	& v_1^1\bigg|_{\subalign{&t^\alpha_p = 0 \\ &(\alpha, p) \not= (1,1)}} = \eta^{1 \mu} \frac{\partial^3 F_0}{\partial t_0^\mu \partial t_0^1 \partial t_0^1}\bigg|_{\subalign{&t^\alpha_p = 0 \\ &(\alpha, p) \not= (1,1)}} = \eta^{1 \mu} \frac{1}{1-t^1_1} \langle \tau_0(e_\mu) \tau_0 (e_1)^2 \rangle_0 = \frac{1}{1 -t_1^1}. \label{eqn:dependencev11}
	\end{align}
	Therefore, from \eqref{eqn:dependencev1} and \eqref{eqn:dependencev11}, we get:
	\begin{align}
	v_1^1 = \frac{1}{1-t^1_1} \left( 1 + \mathcal{O}\left(\frac{t^\cdot_\cdot}{1-t_1^1}\right) \right).
	\end{align}
	Taking $k$ derivatives of \eqref{eqn:vstring}, we have:
	\begin{align}
	v_k^\alpha = \sum_{p \geq 0} \eta^{\alpha \beta} t_{p+1}^\gamma \frac{\partial^{k+2} F_0}{\partial t_p^\gamma \partial t_0^\beta \partial (t_0^1)^k} = \left( \frac{1}{1 - t^1_1} \right)^k \left( \frac{t_k^\alpha}{1-t_1^1} + \mathcal{O}\left(\frac{t^\cdot_\cdot}{1-t_1^1}\right)^2 \right), \qquad (\alpha, k) \not= (1,1).
	\end{align}
	Thus,
	\begin{align}
	\frac{v_k^\alpha}{(v_1^1)^k} = \frac{t_k^\alpha}{1 - t_1^1} + \mathcal{O}\left(\frac{t^\cdot_\cdot}{1-t_1^1}\right)^2.
	\end{align}
	In conclusion, the functions
	\begin{align}
	\Omega^{[g]}_{\alpha, p; \beta, q} (v) := \frac{\partial^2 F_g}{\partial t_p^\alpha \partial t_q^\beta}
	\end{align}
	are of the form $(v_1^1)^{2g} S$, where $S$ is a formal power series in $v^\alpha_p/(v_1^1)^p$. Since we know they only depend on $v, v_1, \dots, v_{3g}$ as a consequence of Proposition \ref{prop:3g-2}, we can conclude
	\begin{eqnarray}
	\Omega^{[g]}_{\alpha, p; \beta, q} = (v_1^1)^{2g} \sum_{k \geq 0} \frac{R_k}{(v_1^1)^k} 
	\end{eqnarray}
	for $R_k$ a differential polynomial depending on $v, v_1, \dots, v_{3g}$ but not on $v_1^1$ of standard degree $k$. Thus, we can write the transformation \eqref{eqn:quasimiuraw} as
	\begin{align}
	w^\alpha = v^\alpha + \eta^{\alpha \beta} \sum_{g=1}^{\infty} \epsilon^{2g} \Omega^{[g]}_{\beta, 0; 1, 0} = v^\alpha + \sum_{g \geq 1} (\epsilon v_1^1)^{2g} \sum_{k \geq 0} \frac{R^\alpha_{g,k}(v, v_1, \dots, v_{3g})}{(v_1^1)^k}, 
	\end{align}
	where $R^\alpha_{g,k}$ is a differential polynomial in $v$ not depending on $v_1^1$ of standard degree $k$. 
\end{proof}

\begin{remark} \label{remark:hamiltonians}
	Another consequence of Proposition \ref{prop:3g-2} is that we can equivalently define the Hamiltonian densities \eqref{eqn:defhamil} in terms of the full two-point correlators $\Omega$, since they differ by a total derivative:
	\begin{align}
	h_{\alpha, p} (w, \dots, w_{3g}) = \Omega_{\alpha, p+1; 1, 0} = \Omega^{[0]}_{\alpha, p+1; 1, 0} + \partial_x \left( \sum_{g=1}^{\infty} \epsilon^{2g} \frac{\partial F_g}{\partial t^\alpha_{p+1}} \right).
	\end{align}
	For the particular case $p = -1$, this means 
	\begin{align}
	w^\alpha = v^\alpha +  \partial_x \left( \eta^{\alpha \beta} \sum_{g=1}^{\infty} \epsilon^{2g} \frac{\partial F_g}{\partial t^\beta_0} \right)
	\end{align}
	so $\int v^\alpha dx = \int w^\alpha dx$.
\end{remark}

We can now write the full hierarchy in Hamiltonian form. The transformation \eqref{eqn:quasimiuraw} induces a transformation on the Poisson bracket $P$ \eqref{eqn:firstbracket} as explained in Example \ref{example:transformationrule}. Explicitly, the deformed Poisson bracket is given by
\begin{align}
A^{\alpha \beta} = \sum_{s \geq 0} A^{\alpha \beta}_s \partial_x^s := \sum_{e, f \geq 0} \frac{\partial w^\alpha}{\partial v^\mu_e} \partial_x^e \circ P^{\mu \nu} \circ (- \partial_x)^f \circ \frac{\partial w^\beta}{\partial v^\nu_f} 
\end{align}
The full hierarchy \eqref{eqn:taucoverfull} is thus given by
\begin{align}
\frac{\partial w^\alpha}{\partial t_q^\beta} = A^{\alpha \gamma} \frac{\delta \bar{h}_{\beta,q}}{\delta w^\gamma}. \label{eqn:fullhierarchy}
\end{align}
Finally, we are ready to state the main result of \cite{BPS12}:
\begin{theorem}[\cite{BPS12}] \label{thm:allispoly}
	\begin{itemize}
		\item The functions $\Omega_{\alpha, p; \beta, q}$ are differential polynomials in $w$, that is, $\Omega_{\alpha, p; \beta, q} = \sum_{g=0}^\infty \epsilon^{2g} \Omega^{g}_{\alpha, p; \beta, q}(w,\dots,w_{3g})$, where $\Omega^{g}_{\alpha, p; \beta, q}$ is a differential polynomial in $w$ of standard degree $2g$ .
		\item The Poisson bracket of the full hierarchy $A^{\alpha \beta} = \sum_{g = 0}^{\infty} \epsilon^{2g} \sum_{s=0}^{2g+1} A_{g,s}^{\alpha \beta} \partial_x^s$ is polynomial in $w$, i.e., the functions $A_{g,s}^{\alpha \beta}$ are differential polynomials in $w$ of standard degree $2g+1-s$.
	\end{itemize}
\end{theorem}

As an immediate consequence of this theorem together with Remark \ref{remark:hamiltonians}, the Hamiltonian densities are polynomial in $w$, and so are the equations of the full hierarchy \eqref{eqn:fullhierarchy}. 

\begin{remark} Note that $\Omega^{0}_{\alpha, p; \beta, q}(w) = \Omega^{[0]}_{\alpha, p; \beta, q}(v)|_{v^\alpha = w^\alpha}$, but the relation between  $\Omega^{g}_{\alpha, p; \beta, q}(w,\dots,w_{3g})$ and  $\Omega^{[g]}_{\alpha, p; \beta, q}(v,\dots,v_{3g})$ is much more involved. 
\end{remark}

\section{The second bracket}  \label{sec:secondbracket}

\subsection{Conformality and bi-Hamiltonian recursion} In Sections \ref{subsec:cohft} and \ref{subsec:prinhierar} we have constructed an integrable hierarchy whose solutions are generated from the partition function of a CohFT. CohFTs are designed to capture the universal basic properties of the Gromov--Witten theories, and it is possible to introduce an extra homogeneity property, which is designed to reflect the computation of the degrees of the Gromov--Witten classes. 

For the purposes of this article, we rewrite the homogeneity condition explicitly as the conformality of the CohFT (as it is done in e.~g.~\cite{alex2020bihamiltonian}), and this extra condition implies the existence of another Hamiltonian structure $K$ of hydrodynamic type compatible with $P$, i.e., $[K,P] = 0$, for which the Hamiltonians $\bar{h}_{\alpha, p}$ are bi-Hamiltonian conserved quantities, i.e., they satisfy $[K, [P, \bar{h}_{\alpha, p}]] = 0$.

\begin{definition}
	A CohFT $\{c_{g,n}\}$ is called conformal if there exist constants $\sfQ_\beta^\alpha$, $\sfB^\alpha$ and $\sfD$ such that $\sfQ_1^\alpha = \delta_1^\alpha$, $\sfQ^\alpha_\beta + \eta^{\alpha\mu} \sfQ^\nu_\mu \eta_{\nu\beta} = (2-\sfD) \delta^\alpha_\beta$, and 
	\begin{align}
	& \left( \frac{1}{2} \deg - (g-1) \sfD - m \right) c_{g,m} (\otimes^m_{i=1} e_{\beta_i}) + \sum_{i=1}^{m} \sfQ^\mu_{\beta_i} c_{g,m} (\otimes_{j=1}^{i-1} e_{\beta_j} \otimes e_\mu \otimes \otimes_{j=i+1}^m e_{\beta_j}  ) \\ \notag  & + \pi_* c_{g,m+1} (\otimes_{i=1}^m e_{\beta_i} \otimes \sfB^\gamma e_\gamma) = 0,
	\end{align}
	where $\deg$ acts on the $k$-th cohomology by multiplication by $k$ and $\pi\colon \oM_{g,m+1}\to \oM_{g,m}$ is the standard map forgetting the last marked point. 
\end{definition}

In terms of the logarithm of the partition function $F=\epsilon^2 \log \tau =\sum_{g=0}^\infty \epsilon^{2g} F_g$, this means
\begin{align}
	 \label{eqn:homogeinity}
& \left( \sum_{d \geq 0} (\sfQ^\gamma_\mu  - d \delta^\gamma_\mu) t^\mu_d\frac{\partial}{\partial t_d^\gamma} + \sfB^\gamma \frac{\partial}{\partial t^\gamma_0} - \sum_{d \geq 0} \eta^{\alpha\mu}\sfB^\beta  \langle \tau_0 (e_\alpha) \tau_0 (e_\beta) \tau_0 (e_\gamma) \rangle_0 t^\gamma_{d+1} \frac{\partial}{\partial t^\mu_d} + \frac{(3-\sfD)}{2} \epsilon \frac{\partial }{\partial \epsilon}  \right) F  \\ \notag
&
= (3 - \sfD) F + \frac{1}{2} \sfB^\gamma  \langle \tau_0 (e_\alpha) \tau_0 (e_\beta) \tau_0 (e_\gamma) \rangle_0 t_0^\alpha t_0^\beta + \epsilon^2 \sfB^\gamma \langle \tau_0 (e_\gamma) \rangle_1  .
\end{align}
It is convenient to introduce notation for a part of this equation which is a vector field on the big phase space:
\begin{align}
	\tsfE\coloneqq \sum_{d \geq 0} (\sfQ^\gamma_\mu  - d \delta^\gamma_\mu) t^\mu_d\frac{\partial}{\partial t_d^\gamma} + \sfB^\gamma \frac{\partial}{\partial t^\gamma_0} - \sum_{d \geq 0} \eta^{\alpha\mu}\sfB^\beta  \langle \tau_0 (e_\alpha) \tau_0 (e_\beta) \tau_0 (e_\gamma) \rangle_0 t^\gamma_{d+1} \frac{\partial}{\partial t^\mu_d}.
\end{align}
Also, denote by $\tilde{R}^\alpha_\beta$ the matrix
\begin{align}
\tilde{R}^\alpha_\beta & \coloneqq\frac{\sfD -1}{2} \delta^\alpha_\beta + \sfQ^\alpha_\beta; 
\\ \notag 
M^\alpha_\beta & \coloneqq \eta^{\alpha\mu} \sfB^\gamma \langle \tau_0 (e_\mu) \tau_0 (e_\beta) \tau_0 (e_\gamma) \rangle_0 .
\end{align}

By direct computation we obtain the following useful lemma that explains the action of $\tsfE$ on the double derivatives of $F$, 
\begin{align}
\Omega^{[g]}_{\alpha,0;\beta,p} & \coloneqq \frac{\partial^2 F_g}{\partial t^\alpha_0\partial t^\beta_p};
&
\Omega^{[g]}_{\alpha,0;\beta,-1} & \coloneqq \eta_{\alpha\beta}\delta_{g,0}.
\end{align}

\begin{lemma} We have: 
	\begin{align}
		\tsfE \Omega^{[g]}_{\alpha,0;\beta,p} +(g(3-\sfD)-1)\Omega^{[g]}_{\alpha,0;\beta,p} + \tilde R^\gamma_\alpha \Omega^{[g]}_{\gamma,0;\beta,p} & =
		 (p+1-\tilde R)^\gamma_\beta  \Omega^{[g]}_{\alpha,0;\gamma,p} +M^\gamma_\beta \Omega^{[g]}_{\alpha,0;\gamma,p-1} ; \label{eqn:EonOmega}
		 \\ \label{eq:MoveR-partialx}
		 \tsfE \partial_x\Omega^{[g]}_{\alpha,0;\beta,p} +g(3-\sfD)\partial_x\Omega^{[g]}_{\alpha,0;\beta,p} + \tilde R^\gamma_\alpha \partial_x\Omega^{[g]}_{\gamma,0;\beta,p} & =
		 (p+1-\tilde R)^\gamma_\beta  \partial_x\Omega^{[g]}_{\alpha,0;\gamma,p} +M^\gamma_\beta \partial_x\Omega^{[g]}_{\alpha,0;\gamma,p-1}.
	\end{align}
\end{lemma}

 Let $F^{\textrm{Frob}} = F_0|_{t^{\bullet}_{\geq 1} = 0}$. As a consequence of the string equation \eqref{eqn:string}, $v^\alpha|_{t^{\bullet}_{\geq 1}=0} = t_0^\alpha$. We have the following system of equations for $F^{\textrm{Frob}}$, derived from the homogeneity \eqref{eqn:homogeinity} and WDVV \eqref{eqn:wddv} equations:
\begin{align}
\frac{\partial^3 F^{\textrm{Frob}}}{\partial v^{\alpha_1} \partial v^{\alpha_2} \partial v^\alpha } \eta^{\alpha \beta} \frac{\partial^3 F^{\textrm{Frob}}}{\partial v^{\alpha_3} \partial v^{\alpha_4} \partial v_0^\beta } = \frac{\partial^3 F^{\textrm{Frob}}}{\partial v^{\alpha_2} \partial v^{\alpha_3} \partial v^\alpha } \eta^{\alpha \beta} \frac{\partial^3 F^{\textrm{Frob}}}{\partial v^{\alpha_1} \partial v^{\alpha_4} \partial v^\beta }, \label{eqn:Frob1} \\
(\sfQ^\gamma_\mu v^\mu + \sfB^\gamma) \frac{\partial}{\partial v^\gamma} F^{\textrm{Frob}} = (3 - \sfD) F^{\textrm{Frob}} + \frac{1}{2} \sfB^\gamma  \langle \tau_0 (e_\alpha) \tau_0 (e_\beta) \tau_0 (e_\gamma) \rangle_0 v^\alpha v^\beta. \label{eqn:Frob2}
\end{align}
The system above realizes the function $F^{\textrm{Frob}}$ as the potential of a Dubrovin--Frobenius manifold, identifying
\begin{align}
E^\alpha = \sfQ^\alpha_\mu v^\mu + \sfB^\alpha
\end{align}
with the coefficients of the Euler vector field. For the theory of Dubrovin--Frobenius manifolds, we refer the reader to \cite{Dubrovin1996, Dubrovin1999}. For the purposes of this article, we only need the following results
\begin{proposition} [\cite{Dubrovin1996}] \label{prop:confisfrob}
	If $F^{\textrm{Frob}} (v)$ satisfies the system \eqref{eqn:Frob1}-\eqref{eqn:Frob2}, there exists a non-degenerate flat metric $g^{\alpha \beta}$ with Christoffel symbols $b^{\alpha \beta}_\gamma$ such that the Poisson operator of hydrodynamic type
	\begin{align}
	K^{\alpha \beta} = g^{\alpha \beta} \partial_x + b^{\alpha \beta}_\gamma v^\gamma_1 \label{eqn:secondbracket}
	\end{align}
	is compatible with $P^{\alpha \beta} = \eta^{\alpha \beta} \partial_x$. Moreover, the explicit expressions of $g$ and $b$ are given by:
	\begin{align}
	g^{\alpha \beta} =\eta^{\alpha \gamma} \eta^{\beta \nu} E^\mu c_{\mu \gamma \nu}, \qquad b^{\alpha \beta}_\gamma = c^{\alpha \delta}_\gamma \tilde{R}^\beta_\delta, \qquad 
	\end{align}
	where $c^{\alpha \beta}_\gamma = \eta^{\alpha \mu} \eta^{\beta \nu} \frac{\partial^3 F^{\textrm{Frob}}}{\partial v^\mu \partial v^\nu \partial v^\gamma}$ are the structure constants of the Frobenius algebra.
\end{proposition}

\begin{remark}
	Under the conformality assumption, the principal hierarchy shown in section \ref{subsec:prinhierar} as constructed in \cite{BPS12} coincides with the Dubrovin--Zhang construction of the principal hierarchy starting from a Dubrovin--Frobenius manifold in \cite{DZ01}, for a particular choice of calibration.
\end{remark}

For the two compatible Poisson brackets, $P$ and $K$, we have a bi-Hamiltonian recursion relation:
\begin{proposition}[see e.~g.~\cite{DZ98}] \label{prop:bihamrec} 
	The following equations hold:
	\begin{align}
	& \{ \cdot, \bar{h}_{\beta, -1} \}_K = \{ \cdot, \bar{h}_{\mu, 0}\}_P (1 - \tilde{R})^\mu_\beta \label{eqn:bihamrec2} \\
	& \{ \cdot,  \bar{h}_{\beta, d} \}_K = \{  \cdot,  \bar{h}_{\mu, d+1} \}_P (d+2 - \tilde{R})^\mu_\beta + \{  \cdot,  \bar{h}_{\mu, d} \}_P M^\mu_\beta, \quad d\geq 0. \label{eqn:bihamrec}
	\end{align}
\end{proposition}
\begin{proof}
	All the arguments in this article are based on the bi-Hamiltonian recursion relation above, which is the central piece to prove the main results of the text. That is why, despite its proof being well-known, we reproduce it here.
	We compute the Hamiltonian vector fields of \eqref{eqn:bihamrec2} term by term: the LHS (after multiplication by $\eta^{\alpha \beta}$) becomes 
	\begin{align}
	g^{\lambda \nu} \partial_x \left( \frac{\delta v^\alpha }{\delta v^\nu} \right) + b^{\lambda \nu}_\gamma v_1^\gamma \frac{\delta v^\alpha }{\delta v^\nu} = b^{\lambda \alpha}_\gamma v_1^\gamma = \tilde{R}^\alpha_\delta c^{\lambda \delta}_\gamma v_1^\gamma
	\end{align}
	On the other hand, the RHS equals:
	\begin{align}
	\eta^{\alpha \beta} \eta^{\lambda \nu} \partial_x \left( \frac{\delta \bar{h}_{\mu, 0}}{\delta v^\nu} \right) (1 - \tilde{R})^\mu_\beta = \eta^{\alpha \beta} \eta^{\lambda \nu} \partial_x \left( \Omega^{[0]}_{\mu, 0; \nu, 0} \right) (1 - \tilde{R})^\mu_\beta \\
	= \eta^{\alpha \beta} \eta^{\lambda \nu} c_{\mu \nu \gamma} v_1^\gamma (1 - \tilde{R})^\mu_\beta = c^\lambda_{\mu \gamma} v_1^\gamma \eta^{\mu \beta} \tilde{R}^\alpha_\beta = \tilde{R}^\alpha_\beta c^{\lambda \beta}_\gamma v_1^\gamma,
	\end{align}
	where we have used that $\tilde{R}^\alpha_\beta \eta^{\beta \gamma} + \tilde{R}^\gamma_\beta \eta^{\alpha \beta} = \eta^{\alpha \gamma}$. Thus, both sides are equal and the relation \eqref{eqn:bihamrec2} holds.
	To prove \eqref{eqn:bihamrec}, first note that
	\begin{align}
	g^{\alpha \mu} \partial_x + b^{\alpha \mu}_\gamma v_1^\gamma = \partial_x \circ g^{\alpha \mu} - b^{\mu \alpha}_\gamma v_1^\gamma 
	\end{align}
	Thus, the LHS equals
	\begin{align}
	\partial_x (g^{\alpha \mu} \Omega^{[0]}_{\beta, d; \mu, 0}) - \partial_x(\tilde{R}^\alpha_\delta \eta^{\delta \theta} \Omega^{[0]}_{\beta, d+1; \theta, 0} )
	\end{align}
	where the second summand comes from the computation
	\begin{align}
	\tilde{R}^\alpha_\delta \eta^{\delta \theta} \partial_x (\Omega^{[0]}_{\beta, d+1; \theta, 0} ) & = \tilde{R}^\alpha_\delta \eta^{\delta \theta} \Omega^{[0]}_{\beta, d; \lambda, 0} \eta^{\lambda \sigma} \partial_x \Omega^{[0]}_{\sigma, 0; \theta, 0} = \tilde{R}^\alpha_\delta \eta^{\delta \theta} \Omega^{[0]}_{\beta, d; \lambda, 0} \eta^{\lambda \sigma} c_{\sigma \theta \gamma} v_1^\gamma \\ \notag 
	& = \tilde{R}^\alpha_\delta c^{\lambda \delta}_\gamma v_1^\gamma \Omega^{[0]}_{\beta,d; \lambda, 0} = b^{\lambda \alpha}_\gamma v_1^\gamma \Omega^{[0]}_{\beta,d; \lambda, 0}.
	\end{align}
	Here we have used TRR-0 \eqref{eqn:TRR0} in the first equality. Thus, equation \eqref{eqn:bihamrec} is equivalent to
	\begin{align}
	g^{\alpha \mu} \Omega^{[0]}_{\beta, d; \mu, 0} - \tilde{R}^\alpha_\delta \eta^{\delta \theta} \Omega^{[0]}_{\beta, d+1; \theta, 0} = \eta^{\alpha \theta} \Omega^{[0]}_{\mu, d+1; \theta, 0} (d+2 - \tilde{R})^\mu_\beta +\eta^{\alpha \theta} \Omega^{[0]}_{\mu, d; \theta, 0} M^\mu_\beta. \label{eqn:equivtobiham}
	\end{align}
	Firstly, using TRR-0 \eqref{eqn:TRR0}, we compute the first summand:
	\begin{align}
	g^{\alpha \mu} \Omega^{[0]}_{\beta,d; \mu, 0} = E^\nu \eta^{\alpha \theta} \frac{\partial}{\partial v^\nu} \Omega^{[0]}_{\theta, 0; \lambda, 0} \eta^{\mu \lambda} \Omega^{[0]}_{\beta, d; \mu, 0} = \eta^{\alpha \theta} E^\mu \frac{\partial}{\partial v^\mu} \Omega^{[0]}_{\beta, d+1; \theta, 0}
	\end{align}
	Secondly, we set $\epsilon = t^{\bullet}_{\geq 1} = 0$ in \eqref{eqn:EonOmega}:
	\begin{align}
	E^\mu \frac{\partial}{\partial v^\mu} \Omega^{[0]}_{\beta, d+1; \theta, 0} - \Omega^{[0]}_{\beta, d+1; \theta, 0} + \tilde{R}^\gamma_\theta \Omega^{[0]}_{\beta, d+1; \gamma, 0} = (d+2-\tilde{R})^\gamma_\beta \Omega^{[0]}_{\gamma, d+1; \theta, 0} + M^\gamma_\beta \Omega^{[0]}_{\gamma, d; \theta, 0}
	\end{align}
	Combining this last equation with the identity $\tilde{R}^\alpha_\beta \eta^{\beta \gamma} + \tilde{R}^\gamma_\beta \eta^{\alpha \beta} = \eta^{\alpha \gamma}$ we see that \eqref{eqn:equivtobiham} holds.
\end{proof}

As before, we can apply the transformation \eqref{eqn:quasimiuraw} to obtain a deformed bracket:
\begin{align}
B^{\alpha \beta} = \sum_{s \geq 0} B^{\alpha \beta}_s \partial_x^s := \sum_{e, f \geq 0} \frac{\partial w^\alpha}{\partial v^\mu_e} \partial_x^e \circ K^{\mu \nu} \circ (- \partial_x)^f \circ \frac{\partial w^\beta}{\partial v^\nu_f}.  
\end{align}
Since \eqref{eqn:quasimiuraw} is a quasi-Miura transformation of $(1,1)$ type by Proposition \ref{prop:type11}, and it only has terms of even degree in $\epsilon$, the second bracket admits an $\epsilon$-expansion of the form
\begin{align}
B^{\alpha \beta} = \sum_{g =0}^\infty \epsilon^{2g} \sum_{s = 0}^{3g+1} B_{g,s}^{\alpha \beta} \partial_x^s, \label{eqn:gexpand}
\end{align}
where $B^{\alpha \beta}_{g,s}$ is a homogeneous differential rational function of type $(1,1)$ and degree $2g+1-s$. Note that the max $s$ for a given $g$ is $3g+1$ as a consequence of Proposition \ref{prop:3g-2}. %Recall Proposition \ref{prop:bihamrec}.

Equations \eqref{eqn:bihamrec2} and \eqref{eqn:bihamrec} can be reformulated in terms of the Hamiltonian vector fields:
\begin{align}
[K, \bar{h}_{\beta, d}] = [P, \bar{h}_{\mu, d+1}] (d +2 - \tilde{R})^\mu_\beta + [P, \bar{h}_{\mu, d}] M^\mu_\beta, \ d \geq -1. \label{eqn:hamvecfield2}
\end{align} Moreover, as a consequence of Proposition \ref{prop:type11}, the Schouten bracket invariance proved in Proposition \ref{prop:changeschouten} still holds and the system \eqref{eqn:hamvecfield2} can be rewritten as
\begin{align}
[B, \bar{h}_{\beta, d}] = [A, \bar{h}_{\mu, d+1}] (d +2 - \tilde{R})^\mu_\beta + [A, \bar{h}_{\mu, d}] M^\mu_\beta, \ d \geq -1. \label{eqn:hamvecfield}
\end{align}
It is illustrative to see how the bi-Hamiltonian recursion relation is preserved in terms of operators instead of the Schouten bracket. Let $L^\alpha_\mu =  \sum_e \frac{\partial w^\alpha}{\partial v^\mu_e} \partial_x^e$ and $(L^*)^\beta_\nu = \sum_f (- \partial_x)^f \circ \frac{\partial w^\beta}{\partial v^\nu_f}$, then:
\begin{align}
B^{\alpha \beta} \frac{\delta}{\delta w^\beta} (\bar{h}_{\gamma, d}) &= L^\alpha_\mu \circ K^{\mu \nu} \circ (L^*)^\beta_\nu \circ \frac{\delta}{\delta w^\beta} (\bar{h}_{\gamma, d}) \label{eqn:hamvecfieldnew} \\ \notag
&= L^\alpha_\mu \circ K^{\mu \nu} \circ \frac{\delta}{\delta v^\nu} (\bar{h}_{\gamma, d})   \\ \notag
&= L^\alpha_\mu \circ P^{\mu \nu} \circ \frac{\delta}{\delta v^\nu} \left( (d+2-\tilde{R})^\lambda_\gamma \bar{h}_{\lambda, d+1}   + M^\lambda_\gamma \bar{h}_{\lambda, d} \right)   \\ \notag
&= L^\alpha_\mu \circ P^{\mu \nu} \circ (L^*)^\beta_\nu \circ \frac{\delta}{\delta w^\beta}  \left( (d+2-\tilde{R})^\lambda_\gamma \bar{h}_{\lambda, d+1}   + M^\lambda_\gamma \bar{h}_{\lambda, d} \right)   \\ \notag
&= A^{\alpha \beta} \frac{\delta}{\delta w^\beta} \left( (d+2-\tilde{R})^\lambda_\gamma \bar{h}_{\lambda, d+1}   + M^\lambda_\gamma \bar{h}_{\lambda, d} \right).  
\end{align}

In other words, $B^{\alpha \beta}$ is an operator of the form \eqref{eqn:gexpand} satisfying \eqref{eqn:hamvecfieldnew}. These are the main equations that we are going to use in the rest of the paper, recycling the same idea thrice.
Firstly, we will show that these two conditions determine $B^{\alpha \beta}$ uniquely. Secondly, we will show that the functions $B^{\alpha \beta}_{g,s}$ are of the form $B^{\alpha \beta}_{g,s} = C^{\alpha \beta}_{g,s}/\det(\eta^{-1}\partial_x\Omega^0)^{n_{g,s}}$, where $n_{g,s} \in \mathbb{Z}_{\geq 0}$ and $C^{\alpha \beta}_{g,s}$ is a differential polynomial in $w$, thus giving an alternative proof to \cite[Theorem 4.2.14]{DZ01} for the second bracket without resorting to the loop equations. Finally, we will prove the vanishing $B^{\alpha \beta}_{g, s} = 0$ for $2g+2 \leq s \leq 3g+1$, which is a necessary condition for $B^{\alpha \beta}$ to be polynomial.

\subsection{Uniqueness theorem}

\begin{theorem} \label{thm:unicity}
	Let $C^{\alpha \beta} = \sum_{g=0}^{\infty} \epsilon^{2g} \sum_{s=0}^{3g + 1} C^{\alpha \beta}_{g,s} \partial_x^s$ be a Poisson operator of type $(1,1)$ in $w$-coordinates satisfying 
	\begin{align}
	[C, \bar{h}_{\beta, d}] = [A, \bar{h}_{\mu, d+1}] (d +2 - \tilde{R})^\mu_\beta + [A, \bar{h}_{\mu, d}] M^\mu_\beta, \ d \geq -1.  
	\end{align}
	Then $C = B$.
\end{theorem}
\begin{proof}
	Let $D = C-B$. Then we have
	\begin{align}
	[D, \bar{h}_{\beta, d}] = 0  ,
	\end{align}
	or, equivalently,
	\begin{align} 
	\sum_{g=0}^{\infty} \epsilon^{2g} \sum_{s=0}^{3g+1} D_{g,s}^{\alpha \gamma} \partial_x^s \left( \frac{\delta \bar{h}_{\beta, d}}{\delta w^\gamma} \right) = 0 \label{eqn:expandunique},
	\end{align}
	for all $1 \leq \beta \leq N$, $d \geq -1$. 
	
	\subsubsection*{Genus $0$}
	Expanding \eqref{eqn:expandunique} in $\epsilon$ and taking the $g=0$ term, we have that 
	\begin{align}
	D^{\alpha \gamma}_{0, 0} \left( \frac{\partial h^{[0]}_{\beta, d}}{\partial w^\gamma} \right) + D^{\alpha \gamma}_{0, 1} \partial_x \left(  \frac{\partial h^{[0]}_{\beta, d}}{\partial w^\gamma} \right) = 0,  
	\end{align}
	where ${h}_{\beta, d} = {h}_{\beta, d}^{[0]} + \epsilon^2 {h}_{\beta, d}^{[1]} + \dots$ is the genus expansion of ${h}_{\beta, d}$ in the $w$ coordinates. Note that we have replaced the variational derivatives $\delta \bar{h}^{[0]}_{\beta, d}/ \delta w^\gamma$ by the partial derivatives $\partial h^{[0]}_{\beta, d}/ \partial w^\gamma$ because the Hamiltonian density $h^{[0]}_{\beta, d}$ does not depend on $w^{\bullet}_{\geq 1}$. For $d = -1$
	\begin{align}
	D^{\alpha \gamma}_{0, 0} \left( \frac{\partial h^{[0]}_{\beta, -1}}{\partial w^\gamma} \right) + D^{\alpha \gamma}_{0, 1} \partial_x \left(  \frac{\partial h^{[0]}_{\beta, -1}}{\partial w^\gamma} \right) = D^{\alpha \gamma}_{0, 0} \eta_{\beta \gamma} = 0  ,
	\end{align}
	so we can conclude that $D^{\alpha \gamma}_{0,0} = 0$. We are left with
	\begin{align}
	D^{\alpha \gamma}_{0, 1} \partial_x \left(  \frac{\partial h^{[0]}_{\beta, d}}{\partial w^\gamma} \right) = 0.  
	\end{align}
	Recall that the dispersionless Hamiltonians $h^{[0]}_{\alpha, d}(w) = h_{\alpha, d}(v)|_{v^\alpha \to w^\alpha}$ satisfy TRR-0 \eqref{eqn:TRR0}, which implies that 
	\begin{align}
	\frac{\partial^2 h^{[0]}_{\beta, d}}{\partial w^\alpha \partial w^\gamma} = \tilde{c}_{\alpha \gamma}^\sigma \frac{ \partial h^{[0]}_{\beta, d-1} }{\partial w^\sigma}, \label{eqn:trr0w}
	\end{align}
	where $\tilde{c}_{\alpha \gamma}^\sigma = \eta^{\sigma \beta} \frac{\partial^3 F^{\textrm{Frob}}}{\partial v^\mu \partial v^\nu \partial v^\gamma} |_ {v^\alpha \to w^\alpha}$.
	Now 
	\begin{align}
	D_{0,1}^{\alpha \gamma} \partial_x \left( \frac{ \partial h^{[0]}_{\beta, 0} }{\partial w^\gamma} \right) = D_{0,1}^{\alpha \gamma} w_1^\mu \tilde{c}^\sigma_{\gamma \mu} \frac{\partial h^{[0]}_{\beta, -1}}{\partial w^\sigma} = D_{0,1}^{\alpha \gamma} w_1^\mu \tilde{c}^\sigma_{\gamma \mu} \eta_{\beta \sigma} = 0   ,
	\end{align}
	so
	\begin{align}
	D_{0,1}^{\alpha \gamma} w^\mu_1 \tilde{c}^{\lambda}_{\gamma \mu} = 0, \quad \lambda = 1, \dots, N
	  .
	\end{align}
	To show that the last equation implies $D_{0, 1}^{\alpha \gamma} = 0$, we need the following lemma:
\begin{lemma} \label{lemma:cancelcv}
	The matrix $\eta^{-1}\partial_x\Omega^0$, written with indices as %$c^\alpha_{\gamma \mu} v^\mu_1$ 
	$\tilde c^\alpha_{\gamma \mu} w^\mu_1 = \eta^{\alpha\beta} \partial_x \Omega^{0}_{\beta,0;\gamma,0}$, is invertible in the $(1,1)$ class. 
\end{lemma}
\begin{proof} The string equation implies that $\tilde c^\alpha_{\gamma \mu} w^\mu_1 - \delta^\alpha_\gamma w^1_1$ is a differential poynomial that does not depend on $w^1_1$.
\end{proof}

\begin{comment}
\begin{proof}
	Let $Q^\gamma$ be a differential rational function of type $(1,1)$ such that $Q^\gamma c^\lambda_{\gamma \mu} v^\mu_1 = 0$. We need to show that $Q^\gamma = 0$. Without loss of generality we can assume that $Q^\gamma$ is homogeneous of degree $d$, and we expand
	\begin{align}
	Q^\gamma = (v_1^1)^d \sum_{l=0}^{\infty} \frac{Q^\gamma_l}{(v_1^1)^l}  
	\end{align}
	with $Q^\gamma_l$ a differential polynomial of degree $l$ not depending on $v_1^1$ and 
	\begin{align}
	c^\lambda_{\gamma \mu} v_1^\mu = c^\lambda_{\gamma 1} v_1^1 + \sum_{\mu \not= 1} c^\lambda_{\gamma \mu} v_1^\mu = \delta^\lambda_{\gamma} v_1^1 + \sum_{\mu \not= 1} c^\lambda_{\gamma \mu} v_1^\mu  .
	\end{align}
	Then
	\begin{align}
	Q^\gamma c^\lambda_{\gamma \mu} v_1^\mu = (v_1^1)^{d+1}  Q_0^\lambda + \sum_{l=0}^{\infty} (v_1^1)^{d-l} \left[ Q_{l + 1}^\lambda + \sum_{\mu \not= 1} c^\lambda_{\gamma \mu} v_1^\mu Q_l^\gamma \right] = 0  
	\end{align}
	which implies
	\begin{align}
	&Q_0^\lambda = 0 \label{eqn:system0} \\ 
	&Q_{l + 1}^\lambda = - \sum_{\mu \not= 1} c^\lambda_{\gamma \mu} v_1^\mu Q_l^\gamma , \qquad \forall \ l \geq 0 \label{eqn:systeml}.
	\end{align}
	In other words, $Q^\gamma = 0$.
\end{proof}
\end{comment}

%The lemma still holds if we replace all the $v$'s by $w$'s, 
So, we can conclude that $D^{\alpha \gamma}_{0, 1} = 0$. 

\subsubsection*{Induction on $g$}
We proceed by induction on $g$. The case $g=0$ has already been proven. Assume $D^{\alpha \beta}_{r, s} = 0$ for all $r \leq g-1$ and for all $s = 0, \dots, 3r+1$. The coefficient of $\epsilon^{2g}$ in \eqref{eqn:expandunique} is
	\begin{align}
	\sum_{r=0}^{g-1} \sum_{s=0}^{3r+1} D_{r,s}^{\alpha \gamma} \partial_x^s \left( \frac{\delta \bar{h}^{[g-r]}_{\beta, d} }{ \delta w^\gamma } \right) + \sum_{s=0}^{3g+1} D_{g,s}^{\alpha \gamma} \partial_x^s \left( \frac{\delta \bar{h}^{[0]}_{\beta, d} }{ \delta w^\gamma } \right) = 0  .
	\end{align} 
	By induction hypothesis, the first summand vanishes, so
	\begin{align}
	\sum_{s=0}^{3g+1} D_{g,s}^{\alpha \gamma} \partial_x^s \left( \frac{\partial h^{[0]}_{\beta, d} }{ \partial w^\gamma } \right) = 0  .
	\end{align}
	Firstly, choosing $d=-1$ implies that $D_{g,0}^{\alpha \gamma} = 0$, so 
	\begin{align}
	\sum_{s=1}^{3g+1} D_{g,s}^{\alpha \gamma} \partial_x^s \left( \frac{\partial h^{[0]}_{\beta, d} }{ \partial w^\gamma } \right) = 0 \label{eqn:expandD}.
	\end{align}
	Secondly, by the chain rule
	\begin{align}
	\partial_x^s \left( \frac{\partial h^{[0]}_{\beta, d} }{ \partial w^\gamma } \right) = \sum_{m=1}^{s} \frac{\partial^{m+1} h^{[0]}_{\beta, d} }{\partial w^\gamma \partial w^{\mu_1} \dots \partial w^{\mu_m}} BP^{\mu_1 \dots \mu_m}_{s, m} \label{eqn:BP},
	\end{align}
	where $BP^{\mu_1 \dots \mu_m}_{s,m}$ is a homogeneous differential polynomial of degree $s$ in the variables $w^\mu_p$, where $\mu = \mu_1 \dots \mu_m$ and $p = 1, \dots, s-m+1$. For this proof, we only need the explicit form of
	\begin{align}
	BP^{\mu_1 \dots \mu_s}_{s,s} = w_1^{\mu_1} w_1^{\mu_2} \dots w_1^{\mu_s} \label{eqn:explicitBP}.
	\end{align}
	Iterating \eqref{eqn:trr0w} $m$ times yields
	\begin{align}
	\frac{\partial^{m+1} h^{[0]}_{\beta, d}  }{\partial w^\gamma \partial w^{\mu_1} \dots \partial w^{\mu_m} } = \sum_{k=1}^{m} BQ^{(m,k), \lambda}_{\gamma \mu_1 \dots \mu_m} \frac{\partial h^{[0]}_{\beta, d-k} }{\partial w^\lambda} \label{eqn:BQ},
	\end{align}
	where $BQ^{(m,k), \lambda}_{\gamma \mu_1 \dots \mu_m}$ is a function in $w$ that can be written in terms of the functions $\tilde{c}^\alpha_{\beta \gamma}$ and their partial derivatives. For this proof, we only need the explicit form of
	\begin{align}
	BQ^{(m,m), \lambda}_{\gamma \mu_1 \dots \mu_m} = \tilde{c}^{\lambda_1}_{\gamma \mu_1} \tilde{c}^{\lambda_2}_{\lambda_1 \mu_2} \dots \tilde{c}^\lambda_{\lambda_{m-1} \mu_m} \label{eqn:explicitBQ}.
	\end{align}
	Inserting \eqref{eqn:BP} and \eqref{eqn:BQ} in \eqref{eqn:expandD}, and changing the order of summation yields
	\begin{align}
	\sum_{k=1}^{3g+1} \left( \sum_{m=k}^{3g+1} \left(  \sum_{s=m}^{3g+1} D_{g,s}^{\alpha \gamma} BP^{\mu_1 \dots \mu_m}_{s,m}   \right) BQ^{(m,k), \lambda}_{\gamma \mu_1 \dots \mu_m} \right) \frac{\partial h^{[0]}_{\beta, d-k}}{\partial w^\lambda} = 0.  
	\end{align}
	Choosing $d=0$ kills all terms except the one with $k=1$, so
	\begin{align}
	\sum_{m=1}^{3g+1} \left(  \sum_{s=m}^{3g+1} D_{g,s}^{\alpha \gamma} BP^{\mu_1 \dots \mu_m}_{s,m}   \right) BQ^{(m,1), \lambda}_{\gamma \mu_1 \dots \mu_m} = 0  ,
	\end{align}
	vanishes and so does
	\begin{align}
	\sum_{k=2}^{3g+1} \left( \sum_{m=k}^{3g+1} \left(  \sum_{s=m}^{3g+1} D_{g,s}^{\alpha \gamma} BP^{\mu_1 \dots \mu_m}_{s,m}   \right) BQ^{(m,k), \lambda}_{\gamma \mu_1 \dots \mu_m} \right) \frac{\partial h^{[0]}_{\beta, d-k}}{\partial w^\lambda} = 0.  
	\end{align}
	Choosing $d=1, 2, \dots 3g$ in the same way shows that
	\begin{align}
	\sum_{m=k}^{3g+1} \left(  \sum_{s=m}^{3g+1} D_{g,s}^{\alpha \gamma} BP^{\mu_1 \dots \mu_m}_{s,m}   \right) BQ^{(m,k), \lambda}_{\gamma \mu_1 \dots \mu_m} = 0, \qquad k=1, \dots, 3g+1. \label{eqn:eachsummand0}
	\end{align}
	Let $k = 3g+1$. By \eqref{eqn:explicitBP} and \eqref{eqn:explicitBQ}, we have
	\begin{align}
	D_{g, 3g+1}^{\alpha \gamma} w_1^{\mu_1} \dots w_1^{\mu_{3g+1}} \tilde{c}_{\gamma \mu_1}^{\lambda_1} \tilde{c}^{\lambda_2}_{\lambda_1 \mu_2} \dots \tilde{c}^\lambda_{\lambda_{3g, \mu_{3g+1}}} = 0.  
	\end{align}
	Regrouping the terms
	\begin{align}
	D^{\alpha \gamma}_{g, 3g+1} w_1^{\mu_1} \tilde{c}^{\lambda_1}_{\gamma \mu_1} w_1^{\mu_2} \tilde{c}^{\lambda_2}_{\lambda_1 \mu_2} \dots w_1^{\mu_{3g}} \tilde{c}^{\lambda_{3g}}_{\lambda_{3g-1}, \mu_{3g}}  w_1^{\mu_{3g+1}} \tilde{c}^\lambda_{\lambda_{3g}, \mu_{3g+1}} = 0.  
	\end{align}
	By Lemma \ref{lemma:cancelcv}, the factor $w_1^{\mu_{3g+1}} \tilde{c}^{\lambda}_{\lambda_{3g}, \mu_{3g+1}}$ can be canceled out, meaning the remaining terms must be zero $D^{\alpha \gamma}_{g, 3g+1} w_1^{\mu_1} \tilde{c}^{\lambda_1}_{\gamma \mu_1} w_1^{\mu_2} \tilde{c}^{\lambda_2}_{\lambda_1 \mu_2} \dots w_1^{\mu_{3g}} \tilde{c}^{\lambda_{3g}}_{\lambda_{3g-1}, \mu_{3g}} =0$. Iterating these cancellations shows that
	\begin{align}
	D^{\alpha \gamma}_{g, 3g+1} = 0  .
	\end{align}
	Replacing this in \eqref{eqn:eachsummand0} yields
	\begin{align}
	\sum_{m=k}^{3g} \left(  \sum_{s=m}^{3g} D_{g,s}^{\alpha \gamma} BP^{\mu_1 \dots \mu_m}_{s,m}   \right) BQ^{(m,k), \lambda}_{\gamma \mu_1 \dots \mu_m} = 0, \qquad k=1, \dots, 3g.  
	\end{align}
	Taking the $k=3g$ term implies, by the same argument as before, that $D^{\alpha \gamma}_{g, 3g} = 0$. Repeating for $k=3g-1, 3g-2, \dots, 1$ shows
	\begin{align}
	D^{\alpha \gamma}_{g, s} = 0, \qquad s = 0, 1, \dots, 3g+1  ,
	\end{align}
	which completes the proof.
\end{proof}

\subsection{Dubrovin--Zhang structural theorem}

An argument based on bi-Hamiltonian recursion as in the proof of Theorem \ref{thm:unicity} is insufficient to show that the functions $B_{g,s}^{\alpha \beta}$ are polynomial. However, it is enough to derive a new proof of the following weaker structural result, which has been already proved in \cite{DZ01} using the loop equation. %The proof presented below is arguably simpler and more illustrative.

\begin{theorem} \label{thm:secondispoly}
	The second Poisson operator of the Dubrovin--Zhang hierarchy $B^{\alpha \beta}$ can be expanded as
	\begin{align}
	B^{\alpha \beta} = \sum_{g=0}^{\infty} \epsilon^{2g} \sum_{s=0}^{3g+1} B^{\alpha \beta}_{g,s} \partial_x^s  ,
	\end{align}
	where $B^{\alpha \beta}_{g,s}$ is a homogeneous differential rational function in the variables $w$ of degree $2g+1-s$ of the form
	\begin{align}
	B^{\alpha \beta}_{g,s} = \frac{C^{\alpha \beta}_{g,s}}{D^{n_{g,s}}},
	\end{align}
	where $D = D (w; w_1) = \det(\tilde{c}^\lambda_{\gamma \mu} w_1^\mu)=\det (\eta^{-1}\partial_x\Omega^0)$, $C^{\alpha \beta}_{g,s}$ is a differential polynomial not divisible by $D$ and $n_{g,s} \in \mathbb{Z}_{\geq 0}$. 
%	Furthermore, the $n_{g,s}$ are upper bounded
%	\begin{align}
%	n_{g,s} \leq \left( \frac{3}{2}g + 1 \right)(g+1).
%	\end{align}
\end{theorem}

\begin{proof}
	Recall $B$ satisfies equation \eqref{eqn:hamvecfieldnew}
	\begin{align}
	B^{\alpha \beta}  \frac{\delta}{\delta w^\beta} (\bar{h}_{\gamma, d}) = A^{\alpha \beta} \frac{\delta}{\delta w^\beta} \left( (d+2-\tilde{R})^\lambda_\gamma \bar{h}_{\lambda, d+1}   + M^\lambda_\gamma \bar{h}_{\lambda, d} \right), \label{eqn:hamvecfieldnew2}
	\end{align}
	whose right hand side is polynomial as a consequence of Theorem \ref{thm:allispoly}. We know that the $g=0$ term of the expansion of $B$ equals $K^{\alpha \beta}|_{v^\alpha \rightarrow w^\alpha}$, which is polynomial, but we will proceed analogously to Theorem~\ref{thm:unicity} even from $g=0$ to illustrate the methods of the proof. Expanding the expression
	\begin{align}
	\sum_{g=0}^{\infty} \epsilon^{2g} \sum_{s=0}^{3g+1} B_{g, s}^{\alpha \gamma} \partial_x^s \left( \frac{\delta \bar{h}_{\beta, d}}{\delta w^\gamma} \right) \label{eqn:expandB}, 
	\end{align}
	which is polynomial, and taking the $g=0$ term implies that
	\begin{align}
	B_{0,0}^{\alpha \gamma} \left( \frac{\partial h_{\beta, d}^{[0]}}{\partial w^\gamma} \right) + B_{0,1}^{\alpha \gamma} \partial_x \left( \frac{\partial h_{\beta, d}^{[0]}}{\partial w^\gamma} \right)  
	\end{align}
	is polynomial. Choosing $d=-1$ shows immediately that $B_{0,0}^{\alpha \gamma}$ is polynomial, so we know
	\begin{align}
	B_{0,1}^{\alpha \gamma} \partial_x \left( \frac{\partial h_{\beta, d}^{[0]}}{\partial w^\gamma} \right)  
	\end{align}
	is polynomial. As in the proof of Theorem \ref{thm:unicity}, we apply a corollary of TRR-0~\eqref{eqn:trr0w} and choose $d=0$ to show that
	\begin{align}
	B_{0,1}^{\alpha \gamma} w^\mu_1 \tilde{c}^\lambda_{\gamma \mu} \label{eqn:polycancel}
	\end{align}
	is polynomial. By Lemma \ref{lemma:cancelcv}, the matrix $(\eta^{-1}\partial_x \Omega^{0})^\lambda_\gamma = w^\mu_1 \tilde{c}^\lambda_{\gamma \mu}$ is invertible. We can write its inverse as $((\eta^{-1}\partial_x \Omega^{0})^{-1})^\gamma_\lambda = \frac{1}{D}T^\gamma_\lambda$, where $T^\gamma_\lambda$ is the transpose of the adjugate of $\eta^{-1}\partial_x \Omega^{0}$, hence a differential polynomial, and $D = \det(w^\mu_1 \tilde{c}^\lambda_{\gamma \mu})$ is its determinant. Therefore, multiplying the polynomial expression \eqref{eqn:polycancel} by $(\eta^{-1}\partial_x \Omega^{0})^{-1}$ implies $B_{0,1}^{\alpha \gamma}$ can be written as the quotient of a differential polynomial by $D$. In other words, we can write
	\begin{align}
	B_{0,1}^{\alpha \gamma} = \frac{C^{\alpha \gamma}_{0,1}}{D^{n_{0,1}}}
	\end{align}
	for $C^{\alpha \gamma}_{0,1}$ a differential polynomial not divisible by $D$ and $n_{0,1} \leq 1$. Assume 
	\begin{align}
	B^{\alpha \beta}_{r,s} = \frac{C^{\alpha \beta}_{r,s}}{D^{n_{r,s}}}
	\end{align}
	with $C^{\alpha \beta}_{r,s}$ a differential polynomial not divisible by $D$ for all $r \leq g-1$ and for all $s\leq 3r+1$. 
	%and $n_{r,s}$ satisfying the upper bounds.
	Let $n = \max_{\substack{0 \leq r \leq g-1 \\ 0 \leq s \leq 3r+1}} (n_{r,s}) %\leq \frac12 (3g-1) g
	$.   
	The coefficient of $\epsilon^{2g}$ in \eqref{eqn:expandB} is 
	\begin{align}
	\sum_{r=0}^{g-1} \sum_{s=0}^{3r+1} B_{r,s}^{\alpha \gamma} \partial_x^s \left( \frac{\delta \bar{h}_{\beta,d}^{[g-r]} }{\delta w^\gamma} \right) + \sum_{s=0}^{3g+1} B_{g,s}^{\alpha \gamma} \partial_x^s \left( \frac{\delta \bar{h}_{\beta, d}^{[0] }}{\delta w^\gamma} \right)  ,
	\end{align}
	which is polynomial. By induction hypothesis, the first summand is a differential polynomial divided by $D^n$, and so is
	\begin{align}
	\sum_{s=0}^{3g+1} B_{g,s}^{\alpha \gamma} \partial_x^s \left( \frac{\partial {h}_{\beta, d}^{[0] }}{\partial w^\gamma} \right) \label{eqn:step1general}.
	\end{align}
	Choosing $d=-1$ implies
	\begin{align}
	B_{g,0}^{\alpha \gamma} = \frac{C_{g,0}^{\alpha \gamma}}{n_{g,0}}
	\end{align}
	with $n_{g,0} \leq n$, so
	\begin{align}
	\sum_{s=1}^{3g+1} B_{g,s}^{\alpha \gamma} \partial_x^s \left( \frac{\partial {h}_{\beta, d}^{[0] }}{\partial w^\gamma} \right) \label{eqn:step2general}
	\end{align}
	can be written as a differential polynomial divided by $D^n$ as well. As in the proof of Theorem \ref{thm:unicity}, we apply iteratively the chain rule \eqref{eqn:BP}, TRR-0 \eqref{eqn:BQ} and choose $d=0,1, \dots, 3g$ to conclude that
	\begin{align}
	\sum_{m=k}^{3g+1} \left( \sum_{s=m}^{3g+1} B_{g,s}^{\alpha \gamma} BP^{\mu_1 \dots \mu_m}_{s,m} \right) BQ^{(m,k), \lambda}_{\gamma \mu_1 \dots \mu_m} \label{eqn:polyfork}
	\end{align}
	is a differential polynomial divided by $D^n$ for all $k = 1, \dots, 3g+1$. Let $k = 3g+1$, then 
	\begin{align}
	B^{\alpha \gamma}_{g, 3g+1} w_1^{\mu_1} \tilde{c}^{\lambda_1}_{\gamma \mu_1} w_1^{\mu_2} \tilde{c}^{\lambda_2}_{\lambda_1 \mu_2} \dots w_1^{\mu_{3g+1}} \tilde{c}^{\lambda}_{\lambda_{3g} \mu_{3g+1}}  = \frac{R^{\alpha\gamma}_{3g,3g+1}}{D^n},
	\end{align}
where $R^{\alpha\gamma}_{3g,3g+1}$ is a differential polynomial. Multiplying this identity by the matrix $((\eta^{-1}\partial_x \Omega^{0})^{-1})^\gamma_\lambda = \frac{1}{D}T^\gamma_\lambda$ from the right $3g+1$ times, we obtain that
	\begin{align}
	B^{\alpha \gamma}_{g, 3g+1} = \frac{C^{\alpha\gamma}_{3g,3g+1}}{D^{n_{g,3g+1}}},
\end{align}
where $C_{3g,3g+1}^{\alpha\gamma}$ is a differential polynomial and $n_{g,3g+1}\leq n+3g+1$.

\begin{comment}
	is a polynomial divided by $D^n$. Since $B^{\alpha \gamma}_{g, 3g+1}$ is of type $(1,1)$, we can apply Lemma \ref{lemma:cancelcv} to invert the last factor and conclude that
	\begin{align}
	B^{\alpha \gamma}_{g, 3g+1} w_1^{\mu_1} \tilde{c}^{\lambda_1}_{\gamma \mu_1} w_1^{\mu_2} \tilde{c}^{\lambda_2}_{\lambda_1 \mu_2} \dots w_1^{\mu_{3g}} \tilde{c}^{\lambda_{3g}}_{\lambda_{3g-1} \mu_{3g}}  
	\end{align}
	is a polynomial divided by $D^{n+1}$. Reiterating shows that
	\begin{align}
	B^{\alpha \gamma}_{g,3g+1} = \frac{C^{\alpha \gamma}_{g,3g+1}}{D^{n_{g,3g+1}}}
	\end{align}
	with $C^{\alpha \gamma}_{g,3g+1}$ a differential polynomial not divisible by $D$ and $n_{g,3g+1} \leq n + 3g + 1$. 
\end{comment}
	
	Taking the $k = 3g$ term in \eqref{eqn:polyfork} shows that
	\begin{align}
	& B^{\alpha \gamma}_{g, 3g} w_1^{\mu_1} \tilde{c}^{\lambda_1}_{\gamma \mu_1} w_1^{\mu_2} \tilde{c}^{\lambda_2}_{\lambda_1 \mu_2} \dots w_1^{\mu_{3g}} \tilde{c}^{\lambda_{3g}}_{\lambda_{3g-1} \mu_{3g}} \\ \notag & + B^{\alpha \gamma}_{g, 3g+1} \left( BP^{\mu_1, \dots, \mu_{3g}}_{3g+1, 3g} BQ^{(3g,3g) \lambda}_{\gamma \mu_1 \dots \mu_{3g}} + BP^{\mu_1 \dots \mu_{3g+1}}_{3g+1, 3g+1} BQ^{(3g+1, 3g) \lambda}_{\gamma, \mu_1, \dots, \mu_{3g}} \right)
	\end{align}
	is a differential polynomial divided by $D^n$.  Therefore,
	\begin{align}
	B^{\alpha \gamma}_{g, 3g} w_1^{\mu_1} \tilde{c}^{\lambda_1}_{\gamma \mu_1} w_1^{\mu_2} \tilde{c}^{\lambda_2}_{\lambda_1 \mu_2} \dots w_1^{\mu_{3g}} \tilde{c}^{\lambda_{3g}}_{\lambda_{3g-1} \mu_{3g}} = \frac{R_{3g,3g}^{\alpha\gamma}}{D^{n+3g+1}},
	\end{align}
	where $R_{3g,3g}^{\alpha\gamma}$ is a differential polynomial. Multiplying this identity by the matrix $((\eta^{-1}\partial_x \Omega^{0})^{-1})^\gamma_\lambda = \frac{1}{D}T^\gamma_\lambda$ from the right $3g$ times, we obtain that
	\begin{align}
	B^{\alpha \gamma}_{g, 3g+1} = \frac{C^{\alpha\gamma}_{3g,3g}}{D^{n_{g,3g}}},
	\end{align}
where $C_{3g,3g}$ is a differential polynomial and $n_{g,3g}\leq n+6g+1$.

Repeating this argument for $k=3g-1, 3g-2, \dots, 1$ shows 
	\begin{align}
	B^{\alpha \gamma}_{g,s} = \frac{C^{\alpha \gamma}_{g,s}}{D^{n_{g,s}}}
	\end{align}
	with $C^{\alpha \gamma}_{g,s}$ a differential polynomial not divisible by $D$.
\end{proof}

\begin{remark}
	It is easy to track through the proof of Theorem~\ref{thm:secondispoly} an estimate for the degrees of the denominators $n_{g,s}$. To make these estimates sharper, one can use the polynomiality in genus $0$ and $1$~\cite{DZ98}, and the result of Theorem \ref{thm:vanishing} below, which states that $B_{g,s}^{\alpha \beta} = 0$ for $s \geq 2g+2$. But, of course, the conjecture of Dubrovin and Zhang suggests that $n_{g,s}=0$.
\end{remark}

\begin{remark}
	The combinatorics of the argument in the proofs of Theorem~\ref{thm:unicity} and Theorem~\ref{thm:secondispoly} basically reflects what happens when one replaces the $\psi$-classes by their pull-backs from the moduli spaces with less number of marked points (cf.~\cite[Equation 3]{BPS12} or~\cite[Proof of Theorem 4]{Liu-Pand}). We make this point precise in Section~\ref{sec:vanishing}.
\end{remark}

Let us also formulate one extra bit of polynomiality of the second Dubrovin--Zhang bracket that follows directly from the proof of Theorem~\ref{thm:secondispoly}:

\begin{theorem} The constant term of the second Poisson operator of the Dubrovin--Zhang hierarchy, $\sum_{g=0}^\infty \epsilon^{2g} B^{\alpha\beta}_{g,0}$, is a differential polynomial.
\end{theorem}

\section{Liu--Pandharipande relations} \label{sec:LiuPand}

\subsection{Relation among the tautological classes}
Fix sets of indices $I_1$ and $I_2$ such that $I_1\sqcup I_2 = \{1,\dots,n\}$. Let $\Delta_{g_1,g_2}\subset \Mgn$ denote a divisor in $\Mgn$ whose generic points are represented by two-component curves intersecting at a node, where the two components have genera $g_1,g_2$ and contain the points with the indices $I_1,I_2$, respectively. 
Note that if $g_i=0$, then $|I_i|$ must be at least $2$, for the stability condition. 

Let $n_1 = |I_1|$, $n_2 = |I_2|$. For each $\Delta_{g_1,g_2}$ we consider the map $\iota_{g_1,g_2}\colon \oM_{g_1,n_1+1}\times  \oM_{g_2,n_2+1}\to\oM_{g,n}$ that glues the last marked points into a node and whose image is $\Delta_{g_1,g_2}$. Let $\psi_{\circ1}$ (respectively, $\psi_{\circ2}$) denote the psi classes at the marked points on the first (respectively, second) component that are glued into the node. 

\begin{proposition}[{\cite[Proposition 1]{Liu-Pand}}] For any $g\geq 0$, $n\geq 4$, $I_1$ and $I_2$ such that $I_1\sqcup I_2 = \{1,\dots,n\}$ and $|I_1|,|I_2|\geq 2$, and an arbitrary $r\geq 0$ we have:
\begin{equation} \label{eq:Liu-Pand-main}
	\sum_{\substack{g_1,g_2\geq 0\\ g_1+g_2=g}}	\sum_{\substack{a_1,a_2\geq 0\\ a_1+a_2=\\ 2g-3+n+r}} (-1)^{a_1} (\iota_{g_1,g_2})_*\psi_{\circ1}^{a_1}\psi_{\circ2}^{a_2} = 0.
\end{equation}
\end{proposition}

Let $n=k+1$, $I_1=\{1,\dots,k-1\}$, $I_2=\{k,k+1\}$, and consider the map $\pi\colon \oM_{g,k+1}\to \oM_{g,k}$ that forgets the last marked point. We apply the push-forward $\pi_*$ to the left hand side of \eqref{eq:Liu-Pand-main} and to the left hand side of \eqref{eq:Liu-Pand-main} multiplied by $\psi_{k+1}$ in order to obtain the following corollaries. 

\begin{corollary} \label{cor:LP1} For any $g\geq 0$, $k\geq 3$, $I_1=\{1,\dots,k-1\}$ and $I_2=\{k\}$, and an arbitrary $r\geq 0$ we have:
	\begin{equation} \label{eq:Liu-Pand-cor1}
		\sum_{\substack{g_1\geq 0, g_2>0\\ g_1+g_2=g}}	\sum_{\substack{a_1,a_2\geq 0\\ a_1+a_2=\\ 2g-2+k+r}} g_2 (-1)^{a_1} (\iota_{g_1,g_2})_*\psi_{\circ1}^{a_1}\psi_{\circ2}^{a_2} = 0.
	\end{equation}
\end{corollary}
	
\begin{corollary} \label{cor:LP2} For any $g\geq 0$, $k\geq 3$, $I_1=\{1,\dots,k-1\}$ and $I_2=\{k\}$, and an arbitrary $r\geq 0$ we have:
	\begin{equation} \label{eq:Liu-Pand-cor2}
		(-1)^{k+r}\psi_k^{2g-2+k+r} +\sum_{\substack{g_1\geq 0, g_2>0\\ g_1+g_2=g}}	\sum_{\substack{a_1,a_2\geq 0\\ a_1+a_2=\\ 2g-3+k+r}} (-1)^{a_1} (\iota_{g_1,g_2})_*\psi_{\circ1}^{a_1}\psi_{\circ2}^{a_2} = 0.
	\end{equation}
\end{corollary}

Taking yet another pushforward, we have the following corollary: 
\begin{corollary}[{\cite[Proposition 2]{Liu-Pand}}] \label{cor:LP3} For any $g\geq 1$, $I_1 = \{1\}$, $I_2 = \{2\}$, and an arbitrary $r\geq 0$ we have:
	\begin{equation} \label{eq:Liu-Pand-cor3}
		-\psi_1^{2g+r} + (-1)^{r}\psi_2^{2g+r} +\sum_{\substack{g_1\geq 0, g_2>0\\ g_1+g_2=g}}	\sum_{\substack{a_1,a_2\geq 0\\ a_1+a_2=\\ 2g-1+r}} (-1)^{a_1} (\iota_{g_1,g_2})_*\psi_{\circ1}^{a_1}\psi_{\circ2}^{a_2} = 0.
	\end{equation}
\end{corollary}

\subsection{Implications for $\partial_x$-derivatives of two-point functions}

Equations~\eqref{eq:Liu-Pand-cor1} and~\eqref{eq:Liu-Pand-cor2} imply a number of identities for the functions $\partial_x^s\Omega^{[g]}_{\alpha,0;\beta,p}$. In order to formulate these identities in a useful way for the computational scheme presented in Section \ref{sec:secondbracket}, we introduce a new notation.

Let $\partial_x^s\Omega^{[g]}_{\alpha,0;\beta,\overline 0}\coloneqq \partial_x^s\Omega^{[g]}_{\alpha,0;\beta,0}$, $s\geq 0$, and for $p\geq 1$ we set
\begin{align}\label{eq:defbar}
	\partial_x^s\Omega^{[g]}_{\alpha,0;\beta,\overline p} \coloneqq	\partial_x^s\Omega^{[g]}_{\alpha,0;\beta,p} - \sum_{q=0}^{p-1} \partial_x^s\Omega^{[g]}_{\alpha,0;\mu,\overline q} \eta^{\mu\nu}  \Omega^{[0]}_{\nu,0;\beta,p-q-1}. 
\end{align}
In other words, in the expansion of $\partial_x^s\Omega^{[g]}_{\alpha,0;\beta,\overline p}$ we use the pull-back of $\psi^p$ from $\oM_{g,s+2}$ at the point with the primary field $\beta$.

\begin{lemma} \label{lem:vanishingwithcoefficient}For $s\geq 1$, $p\geq 2g+s$ we have:
	\begin{align}
		\sum_{\substack{g_1\geq 0, g_2>0\\ g_1+g_2=g}} \sum_{q=0}^{p} g_2(-1)^{q} \partial_x^s\Omega^{[g_1]}_{\alpha,0;\mu,\overline{q}} 
		\eta^{\mu\nu} \Omega^{[g_2]}_{\beta,0;\nu,\overline{p-q}} = 0.
	\end{align} 
\end{lemma}

\begin{lemma} \label{lem:LP-reduction-psi-positiveS}
For $s\geq 1$, $p\geq 2g+s$ we have:
	\begin{align}
			\partial_x^s\Omega^{[g]}_{\alpha,0;\beta,\overline p} = \sum_{\substack{g_1\geq 0, g_2>0\\ g_1+g_2=g}} \sum_{q=0}^{p-1} (-1)^{p-q-1} \partial_x^s\Omega^{[g_1]}_{\alpha,0;\mu,\overline{q}} 
			\eta^{\mu\nu} \Omega^{[g_2]}_{\beta,0;\nu,\overline{p-q-1}}.
	\end{align} 
\end{lemma}

Recall also the notation $\partial_x^s\Omega^{[g]}_{\alpha,0;\beta,-1}\coloneqq \delta_{s,0}\delta_{g,0}\eta_{\alpha\beta}$. We set $\partial_x^s\Omega^{[g]}_{\alpha,0;\beta,\widetilde{-1}}\coloneqq \partial_x^s\Omega^{[g]}_{\alpha,0;\beta,{-1}}$ and for $p\geq 0$
\begin{align}\label{eq:deftilde}
	\partial_x^s\Omega^{[g]}_{\alpha,0;\beta,\widetilde{p}} \coloneqq	\partial_x^s\Omega^{[g]}_{\alpha,0;\beta,p} - \sum_{q=-1}^{p-1} \partial_x^s\Omega^{[g]}_{\alpha,0;\mu,\widetilde q} \eta^{\mu\nu}  \Omega^{[0]}_{\nu,0;\beta,p-q-1}. 
\end{align}
Of course, if $g>0$ or $s>0$, then $\partial_x^s\Omega^{[g]}_{\alpha,0;\beta,\widetilde{p}}=\partial_x^s\Omega^{[g]}_{\alpha,0;\beta,\overline{p}}$, but for $g=s=0$ and $p\geq 0$, $\partial_x^s\Omega^{[g]}_{\alpha,0;\beta,\widetilde{p}}=0$. 
With this extra piece of notation we can include the case $s=0$ in  Lemma~\ref{lem:LP-reduction-psi-positiveS} in the following way:
\begin{lemma} \label{lem:LP-reduction-psi}
	For $s\geq 0$, $p\geq 2g+s$ we have:
	\begin{align}
		\partial_x^s\Omega^{[g]}_{\alpha,0;\beta,\widetilde p} = \sum_{\substack{g_1\geq 0, g_2>0\\ g_1+g_2=g}} \sum_{q=-1}^{p-1} (-1)^{p-q-1} \partial_x^s\Omega^{[g_1]}_{\alpha,0;\mu,\widetilde{q}} 
		\eta^{\mu\nu} \Omega^{[g_2]}_{\beta,0;\nu,\overline{p-q-1}}.
	\end{align} 
\end{lemma}

Lemmata~\ref{lem:vanishingwithcoefficient},~\ref{lem:LP-reduction-psi-positiveS}, and~\ref{lem:LP-reduction-psi} are direct corollaries of Corollaries~\ref{cor:LP1},~\ref{cor:LP2}, and~\ref{cor:LP3}, respectively. It is a rather standard translation of tautological relations into differential equations for the coefficients of the genus expansion of the logarithm of the partition function, see e.~g.~\cite[Proof of Theorem 4]{Liu-Pand}. Another exposition of a detailed step-by-step instruction how one can translate a tautological relation into a PDE is presented in~\cite[Section 2.1.3]{FSZ-rspin}.

\subsection{Variation for $\tsfE$} In fact, as it is explained in~\cite[Proof of Theorem 4]{Liu-Pand}, all lemmata in the previous section work without any change once we replace the operator $\partial_x^s$ with an arbitrary $s$-vector field on the big phase space. The actual result that we use below is a variation of Lemma~\ref{lem:LP-reduction-psi-positiveS} that is related to the vector field $\tsfE$. 

Let $\tsfE\partial_x\Omega^{[g]}_{\alpha,0;\beta,\overline 0}\coloneqq \tsfE\partial_x\Omega^{[g]}_{\alpha,0;\beta,0}$, and for $p\geq 1$ we set
\begin{align}\label{eq:defbarE}
	\tsfE\partial_x\Omega^{[g]}_{\alpha,0;\beta,\overline p} \coloneqq	\tsfE\partial_x\Omega^{[g]}_{\alpha,0;\beta,p} - \sum_{q=0}^{p-1} \tsfE\partial_x\Omega^{[g]}_{\alpha,0;\mu,\overline q} \eta^{\mu\nu}  \Omega^{[0]}_{\nu,0;\beta,p-q-1}. 
\end{align}
In other words, in the expansion of $\tsfE\partial_x\Omega^{[g]}_{\alpha,0;\beta,\overline p}$ we use the pull-back of $\psi^p$ from $\oM_{g,4}$ at the point with the primary field $\beta$.

\begin{lemma} \label{lem:LP-reduction-psi-positiveS-Euler}
	For $p\geq 2g+2$ we have:
	\begin{align}
			\tsfE\partial_x\Omega^{[g]}_{\alpha,0;\beta,\overline p} = \sum_{\substack{g_1\geq 0, g_2>0\\ g_1+g_2=g}} \sum_{q=0}^{p-1} (-1)^{p-q-1} 	\tsfE\partial_x\Omega^{[g_1]}_{\alpha,0;\mu,\overline{q}} 
		\eta^{\mu\nu} \Omega^{[g_2]}_{\beta,0;\nu,\overline{p-q-1}}.
	\end{align} 
\end{lemma}

\section{Vanishing terms of the second bracket} \label{sec:vanishing}

The goal of this section is to prove that all terms of the second Dubrovin--Zhang bracket that have negative degree, and therefore cannot be polynomial, do vanish. The argument goes as follows.

We start with two essential steps to simplify the problem. First, we replace the operator $B$ with a different operator $\tilde B$ that has equivalent vanishing properties but satisfies a simplified version of the bi-Hamiltonian recursion. Second, we employ a triangular structure with respect to the $\epsilon$-degree of the change of variables from $v$ coordinates to $w$ coordinates in order to reduce the problem to the vanishing of the negative degree terms of the operator $\tilde B$ in the $v$ coordinates. 

The latter observation allows us to consider the simplified version of the bi-Hamiltonian recursion in the $v$ coordinates, for which the $\epsilon$-expansion of the $\Omega$-functions has geometric meaning, as it coincides with the expansion in the $t$ variables. This lets us apply various geometric observations from Section~\ref{sec:LiuPand} and homogeneity properties from Section~\ref{sec:secondbracket}  to derive the desired vanishing statement about $\tilde B$.

\subsection{Equivalent form without the variational derivative} \label{sec:PrepB}

Recall the bi-Hamiltonian recursion: 
\begin{equation}
	B^{\alpha\beta}\frac{\delta}{\delta w^\beta} \int \Omega_{1,0;\gamma,p+1} dx = A^{\alpha\beta}\frac{\delta}{\delta w^\beta} \left( \int \Omega_{1,0;\mu,p+2} dx \, (p+2-\tilde R)^\mu_\gamma + \int \Omega_{1,0;\mu,p+1} dx \, M^\mu_\gamma \right)
\end{equation}
for $p\geq -1$. Note that $A^{\alpha\beta} = \partial_x\circ \tilde A^{\alpha\beta}$, where $\tilde A^{\alpha\beta} = \eta^{\alpha\beta}+O(\epsilon^2) = \sum_{g=0}^\infty \epsilon^{2g} \sum_{s=0}^{2g} 
\tilde{A}^{\alpha\beta}_{g,s} \partial_x^s$, where $\tilde{A}^{\alpha\beta}_{g,s}$ are differential polynomials in the coordinates $w$,
%in $w^{\alpha,p}$, $p\geq 1$, with the coefficients in the formal power series in $w^\alpha=w^{\alpha,0}$, 
and the standard gradation of $\tilde{A}^{\alpha\beta}_{g,s}$ is $\deg \tilde{A}^{\alpha\beta}_{g,s} = 2g-s$. Consider the inverse operator, $\tilde A^{-1}_{\alpha\beta}$. It has exactly the same properties as $\tilde A^{\alpha\beta}$, namely, it expands as 
$\tilde A^{-1}_{\alpha\beta} = \eta_{\alpha\beta} +O(\epsilon^2) = \sum_{g=0}^\infty \epsilon^{2g}  \sum_{s=0}^{2g} 
(\tilde A^{-1}_{g,s})_{\alpha\beta}\partial_x^s$, where $(\tilde A^{-1}_{g,s})_{\alpha\beta}$ are differential polynomials, 
%in $w^{\alpha,p}$, $p\geq 1$, with the coefficients in the formal power series in $w^\alpha=w^{\alpha,0}$, 
and the standard gradation of $(\tilde A^{-1}_{g,s})_{\alpha\beta}$ is $\deg (\tilde A^{-1}_{g,s})_{\alpha\beta} = 2g-s$.
Define $\tilde B^{\alpha\beta}\coloneqq B^{\alpha\mu} \tilde A^{-1}_{\mu\nu} \eta^{\nu\beta}$. This operator satisfies a simplified version of bi-Hamiltonian recursion:
\begin{lemma} We have:
\begin{align} \label{eq:BtildeBiHam}
			\tilde B^{\alpha\beta} \Omega_{\beta,0;\gamma,p} = \eta^{\alpha\beta}\partial_x \left( \Omega_{\beta,0;\mu,p+1} \, (p+2-\tilde R)^\mu_\gamma +  \Omega_{\beta,0;\mu,p} \, M^\mu_\gamma \right).
\end{align}
\end{lemma}

\begin{proof} The way the operator $A^{\alpha\beta}$ acts on the variational derivatives of the Hamiltonians implies that 
\begin{align}
	\tilde A^{\alpha\beta}\frac{\delta}{\delta w^\beta} \int \Omega_{1,0;\gamma,p+1} dx = 
	\eta^{\alpha\beta} \Omega_{\beta,0;\gamma,p},	
\end{align}
hence the statement of the lemma. 
\end{proof}

Recall that $B^{\alpha\beta} = \sum_{g=0}^\infty \epsilon^{2g}  \sum_{s=0}^{3g+1} 
B^{\alpha\beta}_{g,s} \partial_x^s$, with $\deg B_{g,s}^{\alpha\beta} = 2g+1-s$. It is easy to see that $\tilde B^{\alpha\beta}$ has expansion with exactly the same properties, namely, $\tilde B^{\alpha\beta} = \sum_{g=0}^\infty \epsilon^{2g} \sum_{s=0}^{3g+1} 
\tilde B^{\alpha\beta}_{g,s} \partial_x^s$, with $\deg \tilde B_{g,s}^{\alpha\beta} = 2g+1-s$. Moreover,

\begin{lemma}\label{lem:B-tildeB} (1) The coefficients of the operator $B^{\alpha\beta}$ are differential polynomials if and only if the coefficients of the operator $\tilde B^{\alpha\beta}$ are differential polynomials. 
	
	(2) The coefficients $B^{\alpha\beta}_{g,s}$, $g\geq 0$, $2g+2\leq s\leq 3g+1$, vanish if and only if the coefficients $\tilde B^{\alpha\beta}_{g,s}$, $g\geq 0$, $2g+2\leq s\leq 3g+1$, vanish. 
\end{lemma}

\begin{proof} Both statements follow from the polynomiality of $\tilde A^{-1}_{\alpha\beta}$.
\end{proof}

Finally, it is a bit easier to work in the $v$ coordinates instead of the $w$ coordinates, but then, of course, all polynomiality properties are destroyed. The vanishing properties are, however, preserved. Namely, consider the expansion of the operator $\tilde B^{\alpha\beta}$ in the $v$ coordinates: $\tilde B^{\alpha\beta} = \sum_{g=0}^\infty \epsilon^{2g} \sum_{s=0}^{3g+1} 
\tilde B^{\alpha\beta}_{[g],s} \partial_x^s$. %Note that $\deg B^{\alpha\beta}_{[g],s} =2g+1-s$. 

\begin{lemma}\label{lem:EquivalenceOfVanishings} The coefficients $\tilde B^{\alpha\beta}_{g,s}$, $g\geq 0$, $2g+2\leq s\leq 3g+1$, vanish if and only if the coefficients $\tilde B^{\alpha\beta}_{[g],s}$, $g\geq 0$, $2g+2\leq s\leq 3g+1$, vanish. 
\end{lemma}

\begin{proof} Indeed, the change of variables from $w$ to $v$ in $\tilde B^{\alpha\beta}(w(v,\epsilon),\epsilon)$ does not affect the terms $\tilde B^{\alpha\beta}_{g,s}$ such that $\tilde B^{\alpha\beta}_{g',s}=0$ for all $g'<g$. More precisely, under this condition $\tilde B^{\alpha\beta}_{[g],s}(v) = \tilde B^{\alpha\beta}_{g,s}(w)|_{w=v}$. The same argument applies also to the change of variables from $v$ to $w$. 
	
Now we proof the lemma by induction on $g$. The base of induction is obvious, and if we prove the equivalence of the vanishings for any $g'<g$, then for any $g'<g$ the top non-vanishing terms in $w$ (respectively, $v$) coordinates are $\tilde B^{\alpha\beta}_{g',s}$ (respectively, $\tilde B^{\alpha\beta}_{[g'],s}$) with $s=2g'+1$. Since $2g'+1<2g+2$ for any $g'<g$, the vanishing of $\tilde B^{\alpha\beta}_{g,s}$, $s\geq 2g+2$ is equivalent to the vanishing of $\tilde B^{\alpha\beta}_{[g],s}$, $s\geq 2g+2$.
\end{proof}

\begin{remark} Assume that the vanishing of $\tilde B^{\alpha\beta}_{g,s}$ (or, equivalently, $\tilde B^{\alpha\beta}_{[g],s}$) is proved for $g\geq 0$, $s\geq 2g+2$. Then the same argument as in the proof of Lemma~\ref{lem:EquivalenceOfVanishings} implies that the polynomiality of $\tilde B^{\alpha\beta}_{g,2g}$ and $\tilde B^{\alpha\beta}_{g,2g+1}$ is equivalent to the polynomiality of $\tilde B^{\alpha\beta}_{[g],2g}$ and $\tilde B^{\alpha\beta}_{[g],2g+1}$, respectively, for any $g\geq 0$.
\end{remark}

\subsection{Vanishing terms} \label{sec:VanishingTerms}

Consider the expansion of the operator $\tilde B^{\alpha\beta}$ in the $v$ coordinates: \begin{align}
	\tilde B^{\alpha\beta} = \sum_{g=0}^\infty \epsilon^{2g} \sum_{s=0}^{3g+1} 
\tilde B^{\alpha\beta}_{[g],s} \partial_x^s.
\end{align}

\begin{proposition} \label{prop:vanishing} We have $\tilde B^{\alpha\beta}_{[g],s} = 0$ for $s\geq 2g+2$, $g\geq 0$. 
\end{proposition}

Before we proceed with the proof of Proposition~\ref{prop:vanishing}, let us note that
 the 	
equation that determines $\tilde B^{\alpha\beta}_{[g],s}$ (once $\tilde B^{\alpha\beta}_{[h],t}$ are known for $h<g$ and for $h=g$, $t>s$), i.e., what corresponds to the $s$-th summand of \eqref{eqn:eachsummand0} in the proof of Theorem \ref{thm:unicity} or, analogously, the $s$-th summand of \eqref{eqn:polyfork} in the proof of Theorem \ref{thm:secondispoly}, can be compactly written as follows:
\begin{lemma} We have:
\begin{align} \label{eq:MainEquationCompact}
	\sum_{\substack{g_1,g_2,t\geq 0\\ g_1+g_2=g}}  \tilde B^{\alpha\beta}_{[g_1],t} \partial_x^t \Omega^{[g_2]}_{\beta,0;\gamma,\widetilde{s-1}} = 
	\eta^{\alpha\beta}\tilde R^\mu_\beta  \partial_x\Omega^{[g]}_{\mu,0;\gamma,\overline{s}} + \tsfE \eta^{\alpha\beta} \partial_x\Omega^{[g]}_{\beta,0;\gamma,\overline{s}} 
	+ g(3-\sfD)  \eta^{\alpha\beta} \partial_x\Omega^{[g]}_{\beta,0;\gamma,\overline{s}}.
\end{align}
\end{lemma}

\begin{proof} Consider the genus $g$ component of Equation~\eqref{eq:BtildeBiHam} with $p=s-1$ (recall that we use the $v$ coordinates for all ingredients of the formula):
\begin{align}
		\sum_{\substack{g_1,g_2,t\geq 0\\ g_1+g_2=g}} 
	\tilde B^{\alpha\beta}_{[g_1],t} \partial_x^t \Omega^{[g_2]}_{\beta,0;\gamma,s-1} = \eta^{\alpha\beta}\partial_x \left( \Omega^{[g]}_{\beta,0;\mu,s} \, (s+1-\tilde R)^\mu_\gamma +  \Omega^{[g]}_{\beta,0;\mu,s-1} \, M^\mu_\gamma \right).
\end{align}
Then we apply~\eqref{eq:MoveR-partialx} with $p=s$. We obtain:
\begin{align}\label{eq:lemmatildebar-main}
	\sum_{\substack{g_1,g_2,t\geq 0\\ g_1+g_2=g}} 
	\tilde B^{\alpha\beta}_{[g_1],t} \partial_x^t \Omega^{[g_2]}_{\beta,0;\gamma,s-1} = 	\eta^{\alpha\beta}\tilde R^\mu_\beta  \partial_x\Omega^{[g]}_{\mu,0;\gamma,{s}} + \tsfE \eta^{\alpha\beta} \partial_x\Omega^{[g]}_{\beta,0;\gamma,{s}} 
	+ g(3-\sfD)  \eta^{\alpha\beta} \partial_x\Omega^{[g]}_{\beta,0;\gamma,{s}}.
\end{align} 
Let us prove that Equation~\eqref{eq:lemmatildebar-main} implies the statement of the lemma by induction on $s$. The base is the case $s=0$; in this case Equations~\eqref{eq:MainEquationCompact} and~\eqref{eq:lemmatildebar-main} are equivalent. Assume the lemma is proved for all $s<S$. Then, for $s=S$ we have:
\begin{align}
	& 
	\sum_{\substack{g_1,g_2,t\geq 0\\ g_1+g_2=g}}  \tilde B^{\alpha\beta}_{[g_1],t} \partial_x^t \Omega^{[g_2]}_{\beta,0;\gamma,\widetilde{S-1}} -
	\eta^{\alpha\beta}\tilde R^\mu_\beta  \partial_x\Omega^{[g]}_{\mu,0;\gamma,\overline{S}} - \tsfE \eta^{\alpha\beta} \partial_x\Omega^{[g]}_{\beta,0;\gamma,\overline{S}} 
	- g(3-\sfD)  \eta^{\alpha\beta} \partial_x\Omega^{[g]}_{\beta,0;\gamma,\overline{S}}
	\\ \notag &
	= 	\sum_{\substack{g_1,g_2,t\geq 0\\ g_1+g_2=g}}  \tilde B^{\alpha\beta}_{[g_1],t} \partial_x^t \Omega^{[g_2]}_{\beta,0;\gamma,{S-1}} -
	\eta^{\alpha\beta}\tilde R^\mu_\beta  \partial_x\Omega^{[g]}_{\mu,0;\gamma,{S}} - \tsfE \eta^{\alpha\beta} \partial_x\Omega^{[g]}_{\beta,0;\gamma,{S}} 
	- g(3-\sfD)  \eta^{\alpha\beta} \partial_x\Omega^{[g]}_{\beta,0;\gamma,{S}}
	\\ \notag & 
	\quad - \sum_{q=0}^{S-1} \Bigg(\sum_{\substack{g_1,g_2,t\geq 0\\ g_1+g_2=g}}  \tilde B^{\alpha\beta}_{[g_1],t} \partial_x^t \Omega^{[g_2]}_{\beta,0;\mu,\widetilde{q-1}} -
	\eta^{\alpha\beta}\tilde R^\mu_\beta  \partial_x\Omega^{[g]}_{\mu,0;\mu,\overline{q}} \\
	\notag & \qquad \qquad \quad - \tsfE \eta^{\alpha\beta} \partial_x\Omega^{[g]}_{\beta,0;\mu,\overline{q}} 
	- g(3-\sfD)  \eta^{\alpha\beta} \partial_x\Omega^{[g]}_{\beta,0;\mu,\overline{q}}
	\Bigg)\eta^{\mu\nu} \Omega^{[0]}_{\nu,0;\gamma,S-q-1}.
\end{align} 
This equality follows from directly from the definitions~\eqref{eq:defbar},~\eqref{eq:deftilde}, and~\eqref{eq:defbarE}. Now, the right hand side of this equality is equal to zero: the first line by Equation~\eqref{eq:lemmatildebar-main}, and the sum over $q$ by the induction assumption. This implies that~\eqref{eq:MainEquationCompact} holds for $s=S$ and completes the inductive proof of the lemma. 
\end{proof}

\begin{proof}[Proof of Proposition~\ref{prop:vanishing}]
We assume that using the computational scheme in Theorems \ref{thm:unicity} and \ref{thm:secondispoly} we already proved by induction the vanishing of $\tilde B^{\alpha\beta}_{[h],t}$ for $h<g$, $t\geq 2h+2$ and for $h=g$, $t>s$. Also, recall that $\tilde B^{\alpha\beta}_{[g_1],t} \partial_x^t \Omega^{[g_2]}_{\beta,0;\gamma,\widetilde{s-1}}=0$ for $0\leq t<s-3g_2$ for dimensional reasons ($\psi_{t+2}^{>3g_2-1+t}$ vanishes on $\oM_{g_2,t+2}$), which we use below for $g_1=g$ and $g_2=0$. Then we have:
\begin{align} \label{eq:LHS-MainEquaCompact}
		& \sum_{\substack{g_1,g_2,t\geq 0\\ g_1+g_2=g }} \tilde B^{\alpha\beta}_{[g_1],t} \partial_x^t \Omega^{[g_2]}_{\beta,0;\gamma,\widetilde{s-1}} = 
	\tilde B^{\alpha\beta}_{[g],s} \partial_x^s \Omega^{[0]}_{\beta,0;\gamma,\widetilde{s-1}} + \sum_{h=0}^{g-1} \sum_{t=0}^{2h+1} \tilde B^{\alpha\beta}_{[h],t} \partial_x^t \Omega^{[g-h]}_{\beta,0;\gamma,\widetilde{s-1}}
	\\ 
	\notag 
	& = 	\tilde B^{\alpha\beta}_{[g],s} \partial_x^s \Omega^{[0]}_{\beta,0;\gamma,\widetilde{s-1}} + \sum_{h=0}^{g-1} \sum_{r=-1}^{s-2} (-1)^{s-r} 
	%\left(
	\sum_{\substack{h_1,h_2\geq 0 \\ h_1+h_2=h}}\sum_{t=0}^{2h_1+1} \tilde B^{\alpha\beta}_{[h_1],t} \partial_x^t \Omega^{[h_2]}_{\beta,0;\mu,\widetilde{r}} 
	%\right) 
	\eta^{\mu\nu} \Omega^{[g-h]}_{\gamma,0;\nu,\overline{s-2-r}}
	\\ \notag
	& = 	\tilde B^{\alpha\beta}_{[g],s} \partial_x^s \Omega^{[0]}_{\beta,0;\gamma,\widetilde{s-1}} + \sum_{h=0}^{g-1} \sum_{r=-1}^{s-2} (-1)^{s-r} 
	\Bigg(
	\eta^{\alpha\beta}\tilde R^\xi_\beta  \partial_x\Omega^{[h]}_{\xi,0;\mu,\overline{r+1}} + \tsfE \eta^{\alpha\beta} \partial_x\Omega^{[h]}_{\beta,0;\mu,\overline{r+1}} 
	\\ \notag  & \phantom{=\ }
	+ h(3-\sfD)  \eta^{\alpha\beta} \partial_x\Omega^{[h]}_{\beta,0;\mu,\overline{r+1}}
	\Bigg) 
	\eta^{\mu\nu} \Omega^{[g-h]}_{\gamma,0;\nu,\overline{s-2-r}}.
\end{align}
Here for the second equality we use Lemma~\ref{lem:LP-reduction-psi}, and for the third equality we use Equation~\eqref{eq:MainEquationCompact} for $g=h$ and $s=r+1$. Note that the condition of Lemma~\ref{lem:LP-reduction-psi} is indeed satisfied: since $t\leq 2h+1$ and $2g+2\leq s$, we have $2(g-h)+t\leq s-1$.

On the other hand, for $s\geq 2g+2$ (this inequality is crucially important for the second summand, for the first and the third ones $s\geq 2g+1$ would be sufficient), we have from Lemmata~\ref{lem:LP-reduction-psi-positiveS} and~\ref{lem:LP-reduction-psi-positiveS-Euler} the following:
\begin{align} \label{eq:RHSMainEquaCompact}
	& \eta^{\alpha\beta}\tilde R^\mu_\beta  \partial_x\Omega^{[g]}_{\mu,0;\gamma,\overline{s}} + \tsfE \eta^{\alpha\beta} \partial_x\Omega^{[g]}_{\beta,0;\gamma,\overline{s}} 
	+ g(3-\sfD)  \eta^{\alpha\beta} \partial_x\Omega^{[g]}_{\beta,0;\gamma,\overline{s}}
	\\ \notag
	& =\sum_{h=0}^{g-1} \sum_{r=0}^{s-1} (-1)^{s-r-1} 	\Bigg(
	\eta^{\alpha\beta}\tilde R^\xi_\beta  \partial_x\Omega^{[h]}_{\xi,0;\mu,\overline{r}} + \tsfE \eta^{\alpha\beta} \partial_x\Omega^{[h]}_{\beta,0;\mu,\overline{r}} 
	+ g(3-\sfD)  \eta^{\alpha\beta} \partial_x\Omega^{[h]}_{\beta,0;\mu,\overline{r}}
	\Bigg) 
	\eta^{\mu\nu} \Omega^{[g-h]}_{\gamma,0;\nu,\overline{s-1-r}}.
\end{align}
Substituting this expression and Equation~\eqref{eq:LHS-MainEquaCompact} in Equation~\eqref{eq:MainEquationCompact}, we obtain:
\begin{align}\label{eq:1}
	\tilde B^{\alpha\beta}_{[g],s} \partial_x^s \Omega^{[0]}_{\beta,0;\gamma,\widetilde{s-1}}
	= (3-\sfD) \eta^{\alpha\beta} \sum_{h=0}^{g-1} \sum_{r=0}^{s-1} (-1)^{s-r-1} 
	(g-h)\partial_x\Omega^{[h]}_{\beta,0;\mu,\overline{r}}
	\eta^{\mu\nu} \Omega^{[g-h]}_{\gamma,0;\nu,\overline{s-1-r}}.
\end{align}
For $s-1\geq 2g+1$ the right hand side of this equation is equal to zero by Lemma~\ref{lem:vanishingwithcoefficient}. Note that 
\begin{align}\label{eq:2}
\partial_x^s \Omega^{[0]}_{\beta,0;\gamma,\widetilde{s-1}} = \delta_{\beta}^{\mu_1}\delta_\gamma^{\nu_s} \prod_{i=1}^s \partial_x  \Omega^{[0]}_{\mu_i,0;\nu_i,0} \prod_{i=1}^{s-1} \eta^{\nu_i\mu_{i+1}}.
\end{align}
(we prove it below in Lemma~\ref{lem:eqn0}).

Now, since $\partial_x^s \Omega^{[0]}_{\beta,0;\gamma,\widetilde{s-1}}$ as given by~\eqref{eq:2} is invertible, we can use the vanishing of the right hand side of Equation~\eqref{eq:1} to conclude that $\tilde B^{\alpha\beta}_{[g],s}=0$. 
\end{proof}

\begin{lemma}\label{lem:eqn0} Equation~\eqref{eq:2} holds for $s\geq 0$.
\end{lemma}

\begin{proof} We prove it by induction on $s$. The full induction statement is the following. For any $s\geq 0$
\begin{align}\label{eq:2bis}
		\partial_x^t \Omega^{[0]}_{\beta,0;\gamma,\widetilde{s-1}} =
		\begin{cases}
		 \delta_{\beta}^{\mu_1}\delta_\gamma^{\nu_s} \prod_{i=1}^s \partial_x  \Omega^{[0]}_{\mu_i,0;\nu_i,0} \prod_{i=1}^{s-1} \eta^{\nu_i\mu_{i+1}} & t= s; \\
		 0 & 0\leq t < s.			
		\end{cases}		
\end{align}
For $s=0$ it is the definition of $\Omega^{[0]}_{\beta,0;\gamma,\widetilde{-1}}$. For the induction step we have to recall the topological recursion relation in genus $0$~\eqref{eqn:TRR0}, which implies that for $p\geq 0$ 
	\begin{align}
		\partial_x \Omega^{[0]}_{\alpha,0;\beta,p} = \partial_x \Omega^{[0]}_{\alpha,0;\mu,0}\eta^{\mu\nu}  \Omega^{[0]}_{\nu,0;\beta,p-1},
	\end{align}
and, therefore,
\begin{align} \label{eq:TRRtilde}
			\partial_x \Omega^{[0]}_{\alpha,0;\beta,\widetilde {p}} = \partial_x \Omega^{[0]}_{\alpha,0;\mu,0}\eta^{\mu\nu}  \Omega^{[0]}_{\nu,0;\beta,\widetilde{p-1}}; 
	\end{align}
Assume~\eqref{eq:2bis} is proved for $s\leq S$. Then for $t=0$ and $s=S+1$ the required vanishing follows directly from the induction assumption applied to the right hand side of Equation~\eqref{eq:deftilde}. If $t\geq 1$, then for $s=S+1$ we use ~\eqref{eq:TRRtilde} to obtain
\begin{align}
	& \partial_x^t \Omega^{[0]}_{\beta,0;\gamma,\widetilde{S}} = 
		\partial_x^{t-1}\left( \partial_x\Omega^{[0]}_{\beta,0;\xi,0}\eta^{\xi\zeta}\Omega^{[0]}_{\zeta,0;\gamma,\widetilde{S-1}}\right) 
	= \sum_{u=0}^{t-1} \binom{t-1}{u} \partial_x^{u+1}\Omega^{[0]}_{\beta,0;\xi,0}\eta^{\xi\zeta} \partial_x^{t-1-u}\Omega^{[0]}_{\zeta,0;\gamma,\widetilde{S-1}}.
\end{align}
If $1\leq t\leq S$, then $t-1-u<S$ for any $u=0,\dots,t-1$, and then this expression is equal to zero by the induction assumption. Let $t=S+1$. Then $t-1-u = S-u < S$ for $u=1,\dots,t-1$, and therefore the corresponding summands are equal to zero by the induction assumption. Thus for $t=S+1$ we have
\begin{equation}
	\partial_x^{S+1} \Omega^{[0]}_{\beta,0;\gamma,\widetilde{S}} = \partial_x\Omega^{[0]}_{\beta,0;\xi,0}\eta^{\xi\zeta} \partial_x^{S} \Omega^{[0]}_{\zeta,0;\gamma,\widetilde{S-1}}
\end{equation}
Substitution of the non-vanishing case of the Equation~\eqref{eq:2bis} for $s=S$ into this formula proves the non-vanishing case of the Equation~\eqref{eq:2bis} for $s=S+1$, which completes the step of induction and proves the lemma.
\end{proof}

Now we are ready to state and prove our main theorem, which appears to be a direct corollary of Proposition~\ref{prop:vanishing}.

\begin{theorem} Consider the $\epsilon$-expansion of the second Dubrovin--Zhang bracket in the coordinates $w$: $B^{\alpha\beta} = \sum_{g=0}^\infty \epsilon^{2g} \sum_{s=0}^{3g+1} 
	B^{\alpha\beta}_{g,s} \partial_x^s$. We have $B^{\alpha\beta}_{g,s}=0$ for $g\geq 0$, $s\geq 2g+2$. \label{thm:vanishing}
\end{theorem}

\begin{proof} Lemmata~\ref{lem:B-tildeB} and~\ref{lem:EquivalenceOfVanishings} imply that the statement of the theorem is equivalent to the statement of Proposition~\ref{prop:vanishing}.
\end{proof}

\bibliography{LPrelVanishSB-20220106-s}
\bibliographystyle{alpha}

\end{document}